\documentclass[journal,comsoc]{IEEEtran}

\usepackage{cite}	
\usepackage{graphicx} 
\usepackage[latin1]{inputenc} 
\usepackage[T1]{fontenc} 
\usepackage{amsmath,amsfonts,amsbsy,amssymb} 

\usepackage{mathabx} 
\usepackage{amssymb} 
\usepackage{amsmath}
\usepackage{amsthm}	
\usepackage{mathrsfs}
\usepackage[nolist]{acronym} 
\usepackage{tabularx} 
\usepackage{multirow}
\usepackage{wasysym}
\usepackage{cancel}
\usepackage{float}
\usepackage{color} 
\usepackage[lined, boxed, linesnumbered, ruled]{algorithm2e}

\usepackage{enumitem} 

\newtheorem{lemma}{Lemma}

\newtheorem{theorem}{Theorem}
\newtheorem{remark}{Remark}

\begin{document}
\title{Effective Energy Efficiency of Ultra-reliable Low Latency Communication} 

\author{
	\IEEEauthorblockN{Mohammad Shehab, Hirley Alves, Eduard A. Jorswieck, Endrit Dosti, and Matti Latva-aho}
	\thanks{M. Shehab, H. Alves, Endrit Dosti, and M. Latva-aho are with Centre for Wireless Communications (CWC), University of Oulu, Finland. Email: firstname.lastname@oulu.fi. Eduard A. Jorswieck is with the Department of Information Theory and Communication Systems, Technische Universtät Braunschweig, Germany, Email: e.jorswieck@tu-bs.de. E. Dosti is with the Department of Signal Processing and Acoustics, Aalto University, Finland. Email: endrit.dosti@aalto.fi} 
	\thanks{This work is partially supported by Academy of Finland 6Genesis Flagship (Grant no. 318927), Aka Project EE-IoT (Grant no. 319008). The work of E. Jorswieck is partly funded by the German Research Foundation (DFG, Deutsche Forschungsgemeinschaft) as part of Germany's Excellence Strategy - EXC 2050/1 - Project ID 390696704 - Cluster of Excellence "Centre for Tactile Internet with Human-in-the-Loop" (CeTI) of Technische Universit{\"a}t Dresden.}
}

%
\maketitle

\vspace{-7mm}
\begin{abstract}
Effective Capacity defines the maximum communication rate subject to a specific delay constraint, while effective energy efficiency (EEE) indicates the ratio between effective capacity and power consumption. We analyze the EEE of ultra-reliable networks operating in the finite blocklength regime. We obtain a closed form approximation for the EEE in quasi-static Nakagami-$m$ (and Rayleigh as sub-case) fading channels as a function of power, error probability, and latency. Furthermore, we characterize the QoS constrained EEE maximization problem for different power consumption models, which shows a significant difference between finite and infinite blocklength coding with respect to EEE and optimal power allocation strategy. As asserted in the literature, achieving ultra-reliability using one transmission consumes huge amount of power, which is not applicable for energy limited IoT devices. In this context, accounting for empty buffer probability in machine type communication (MTC) and extending the maximum delay tolerance jointly enhances the EEE and allows for adaptive retransmission of faulty packets. Our analysis reveals that obtaining the optimum error probability for each transmission by minimizing the non-empty buffer probability approaches EEE optimality, while being analytically tractable via Dinkelbach's algorithm. Furthermore, the results illustrate the power saving and the significant EEE gain attained by applying adaptive retransmission protocols, while sacrificing a limited increase in latency.
\end{abstract}


\section{Introduction}\label{introduction}
The new generations of mobile communication are expected to support a multitude of smart devices interconnected via machine type networks, enabling the Internet of Things (IoT). Energy efficient transmission while guaranteeing quality-of-service (QoS) is an ultimate goal in the design of future wireless networks. QoS constraints ranging from low latency in the order of few milliseconds and packet loss rate ($< 10^{-4}$) are key requirements for \textit{Ultra-Reliable Low Latency Communication} (URLLC) \cite{PopovskiURLLC2019}. In order to boost throughput and reliability while guaranteeing low latency, it becomes crucial to investigate and optimize the resources that are allocated for transmission. In most cases, URLLC devices such remote sensors have limited energy resources which dictates careful planning of throughput maximization with wise energy consumption models \cite{green_IoT}. Furthermore, besides possible health risks of electromagnetic radiation in over populated areas such as city centers \cite{Kurnaz}, the information and communication technology industry is projected to contribute to 6$\%$ of global CO$_2$ emission in 2020 \cite{green}. This urges the invention of low power consumption, green communication schemes, yet able to perform with QoS guarantees.

In order to satisfy extremely low latency in real time applications and emerging technologies such as e-health, industrial IoT, and autonomous vehicles, an attractive solution is communication with short messages \cite{paper3}. Therein, the lengths of the packets to be communicated are short, but their importance is extremely high. When the packets are short and delay requirements are stringent, performance metrics, such as Shannon capacity or outage capacity, provide a poor benchmark, and therefore, fundamentally new approaches are needed \cite {paper5,short_ur}. In this context, the maximum achievable rate of finite blocklength packets was defined in \cite{polyanskiy} as a function of blocklength and error probability.

As envisaged by \cite{mehdi}, the design of URLLC focuses on the tail distributions of reliability and latency instead of average metrics. Here arises the challenge of how to incorporate energy efficiency with the data rates, delay, and reliability requirements imposed by the International Telecommunication Union (ITU) and for MTC towards 6G. In this sense, metrics such as effective capacity (EC) and effective energy efficiency (EEE) are meant to capture tail statistical delay requirements in parallel with transmission throughput. 

\subsection{Related Work}
The effective capacity metric was first introduced in \cite{paper6} to guarantee statistical QoS requirements by capturing the physical and link layers aspects. It maps the maximum arrival rate that can be supported by a network with a maximum delay bound of $\delta$ and a delay outage probability. Unlike the the Delay-Sensitive Area Spectral Efficiency metric which only accounts for the transmission delay \cite{DSA}, the EC accounts for the statistical QoS aspect in terms of delay outage probability and the maximum delay bound. In \cite{paper5}, Gursoy characterized the EC in bits per channel use (bpcu) for short packets in quasi-static fading channels where the channel coefficients remain constant for the whole time spanning one packet transmission. Moreover, in \cite{Gursoy2}, Gursoy et al. extended their analysis to multiple users but without considering the power consumption and energy efficiency aspects. Meanwhile, the per-node EC in massive MTC networks was studied in \cite{eucnc} proposing three methods to alleviate interference namely power control, graceful degradation of delay constraint and the hybrid method. 

Effective energy efficiency is defined as the ratio between effective capacity and the total power consumption. In this sense, EEE captures the interplay between power, delay, and reliability and thus, fits well for dealing with the inherent energy-limited and bursty traffic scenarios in MTC characterized by URLLC. The maximization of EEE is of great importance for green IoT, where the goal is to maximize the throughput for each consumed unit of power. The EEE can be used as a measure of how efficient an IoT system is in terms of power consumption and energy saving. This is a very useful metric in the study and design of remote IoT devices that run on installed batteries with limited energy supply and are required to communicate with URLLC requirements. We refer to the novel factors that affect energy efficiency in MTC, which are strict delay and error constraints, bursty traffic, empty buffer probability, and communication on finite blocklength packets. In \cite{paper7}, the empty buffer probability (EBP) model was considered as an EEE booster for long packets transmission. The energy efficiency gap which results from utilizing finite blocklength packets and the optimum power allocation in this case were characterized in \cite{iswcs_eee} for Rayleigh fading channel. The trade-off between EEE and EC was studied in \cite{paper8}, where the authors suggested an algorithm to maximize the EC subject to EEE constraint. However, the probability of transmission error that appears in finite blocklength communication due to imperfect coding was not considered. The authors of \cite{paper20} showed that the relation between EEE and delay in wireless systems is not always a trade-off. They concluded that a linear relation between service rate and power consumption leads to an EEE-delay non-trade-off region.

On the other hand, it has been well-established that achieving ultra reliability requires the utilization of diversity schemes such as ARQ retransmission protocols. The utilization of this family of protocols has been embraced by several up to date systems such as 5G NR \cite{3GPP}. In this context, the authors of \cite{EC_rtx} discussed the EC of ARQ schemes for matrix-exponential distributed fading channels. They suggested that exploiting spatial diversity as the case in MIMO would reduce the sensitivity of EC to variations in the delay exponent $\theta$. However, only the work in \cite{ARQ_FB} studied the finite blocklength effect in ARQ Assisted URLLLC but without accounting power consumption as in the EEE metric.

\subsection{Contributions} 
In this work, we build upon \cite{iswcs_eee} and propose a finite blocklength model for the EEE in quasi-static Nakagami-$m$ fading assuming a linear power consumption model. We characterize the optimum power allocation strategy that maximizes the EEE. We account for the EBP and in power consumption model  and prove that this model is valid for short packets. Our analysis indicates that considering EBP with short packets allows for a more precise characterization of the EC and EEE. Results show that higher EC and EEE are obtained under EBP in contrast to the full buffer scenario. Besides, we illustrate the EEE gap between infinite and the finite blocklength models, and highlight the effect of network congestion on the EEE performance. We analyse how the optimum power allocation is affected by limiting the packet length and the performance gap that appears accordingly for different types of fading. In addition, we evaluate the trade-off between the EEE and the delay tolerance.

A key contribution of this work is that we exploit the non-EBP model by incorporating retransmission of faulty packets in the instants when the buffer is empty, which renders an EBP-ARQ scheme. We evaluate the EEE of the proposed  scheme and also compare this case to the basic EBP model. The results show that this scenario grants a significant improvement in the EEE and reduction in the power consumption. Our analysis characterizes the optimum transmission error probability for each transmission so that the global EEE is maximized and the power consumption is minimized. The average latency is also analyzed for this scenario where the results indicate a very limited increase in latency as a cost of the high gain in EEE.  

The contributions of this work are summarized as follows:
\begin{itemize}
	\item We derive closed form approximation for the EEE under quasi-static fading.
	\item We characterize the optimum power allocation that maximizes the EEE\footnote{It is worth noticing that in \cite{iswcs_eee} we evaluated only numerically the optimal power allocation that maximizes the EEE under Rayleigh fading and without any assumptions on the EBP.} for linear and EBP power consumption models and highlight the trade-off between EEE and latency.  
	\item We show the performance and optimal power allocation gap that results from utilizing short packets when compared to finite blocklength. Unlike \cite{TVT_Shehab}, where we maximize effective capacity via optimal power allocation under Rayleigh fading while neglecting energy consumption of the devices, herein we analyze the optimal power allocation for different Nakagami-$m$ fading setups.
	\item We present a basic framework for applying ARQ by considering retransmission of faulty packets in the EBP model. A buffer aware strategy is adopted to allow for rate adaption when the buffer is empty in the time slot following the current transmission. The proposed framework allows for a significant rate gain and renders a significant boost in the EEE, while maintaining a limited increase in average latency.
	\item We solve the optimization problem for optimum error allocation at each transmission with a target reliability constraint via low complexity Dinkelbach's algorithm. The solution shows that minimizing the non-EBP jointly reduces the power consumption and maximizes the EEE.
\end{itemize}

\subsection{Outline}
The rest of the paper is organized as follows: in Section \ref {system model}, we introduce the system model and clarify the relation between EC and EEE. Next, Section \ref{EEE_finite blocklength} presents the EEE analysis in Nakagami-$m$ quasi-static fading and characterizes optimum error and power allocation for the linear power consumption model. We illustrate the EBP model in finite blocklength transmission and characterize the EEE maximization with reliability, latency, and power consumption constraints in Section \ref{EBP}. After that, Section \ref{retransmission} studies the retransmission of faulty packets in the empty buffer instants, while the results are discussed in Section \ref{results}. Finally, Section \ref{conclusion} concludes the paper. To make the paper more tractable, we summarize the key abbreviations and symbols that will appear throughout the paper in Table \ref{t1}.
\begin{table}[!t]
	\centering
	\caption{Important abbreviations and symbols.}
	\label{t1}
	\renewcommand{\arraystretch}{1.3}
	\begin{tabular}{p{1.5cm} p{5cm}}
		\hline

		bpcu & bits per channel use \\
		max & maximize  \\
		NBP & non-empty buffer probability \\
		s.t & subject to    \\
		\vspace{8mm} \\

		$m$ & fading parameter \\
		$n$ & blocklength \\
		$P_t$ & power consumption \\
		$p_{nb}$ & non-empty buffer probability  \\
		$r$ & normalized achievable rate \\
		
		\vspace{8mm} \\
		 
	    $C_e$ & effective capacity \\
	    $\delta$ & maximum delay  \\
	    $\mathbb{E}[ \ ]$ & expectation of \\
	    $\Lambda$ & delay outage probability   \\
	    $Q(x)$ & Gaussian Q-function

	    \vspace{8mm} \\
	    $\theta$ & delay exponent \\
	    $\epsilon_t$ & target error probability \\
	    $\epsilon^*$ & optimum error probability \\
	    $\eta_{ee}$ & effective energy efficiency\\
	    $\lambda$ & arrival rate \\
	    $\rho$ & average signal-to-noise ratio \\
	    $\zeta$ & inverse drain efficiency \\
	 \hline
	\end{tabular} 
\end{table}
\section{System model and preliminaries} \label{system model}
We consider a communication scenario in which an energy-limited sensor transmits data to a common aggregator through a quasi-static Nakagami-$m$ fading channel. The received vector $\mathbf {y}\in \mathbb{C}^n$ is
\begin{align}\label{eq1}
\mathbf {y}=h\mathbf {x}+\mathbf {w},
\end{align}
where $\mathbf {x} \in \mathbb{C}^n$ is the transmitted packet, and the block flat-fading coefficient is denoted by $h \in \mathbb{C}$ which is assumed to be independent and identically distributed (i.i.d). This implies that $h$ remains constant over the blocklength $n$, but changes from block to block. The  blocklength is assumed to be smaller than the channel's coherence time. Lastly, $\mathbf{w}$ is the additive complex Gaussian noise vector whose entries are of unit variance. We assume that CSI is available at each node. 
%
Note that CSI acquisition at the transmitter in the URLLC setup is  feasible  whenever the channel state remains constant over multiple symbols\footnote{Channel estimation is a cumbersome task for conventional systems and also under URLLC constraints. However, estimates can be reliably acquired via feedback channel in FDD system, or by exploiting channel reciprocity in TDD in line with channel inversion power control methods. CSI is obtained at reception if the latency constraint allows additional overhead due to pilot based transmissions, or alternatively via non-coherent transmissions \cite{PopovskiURLLC2019}.}. In most communication environments, the channel coherence time is much larger than the URLLC transmissions in mini-slots of duration 0.1 ms, and thus spans of multiple TTI. This gives the transmitter sufficient time to perfectly acquire CSI \cite{CSI_Onel}. Recent machine learning tools facilitate this task specially if the channel coefficients are highly correlated within a short period of time. In this case, the transmitter node can exploit a recently received signal from the other node or request a training sequence in order to estimate the channel \cite{CSIT}. Additionally, as in \cite{paper3}, we aim to provide a performance benchmark for energy efficiency of these networks, where the effect of imperfect CSI is beyond the scope of our work.

\subsection{Communication at Finite Blocklength}
In finite blocklength transmission, unlike Shannon's model, short packets are conveyed at rate that depends on the blocklength $n$ and the packet error probability $\epsilon \in\left[ 0,1\right]$, which is small but not vanishing. The normalized achievable rate, in (bpcu), is \cite{paper5}
\begin{align}\label{eq3}
\begin{split}
r(\rho)\!\approx&\!\log_2(1\!+\!\rho|h|^2) 
\!-\!\frac{Q^{-1}
	(\epsilon)\log_2(e)}{\sqrt{n}}\!\sqrt{1\!-\!\frac{1}{\left(1\!+\!\rho|h|^2\right)^{2}} } ,	
\end{split}
\end{align}
where $Q(\cdot)=\int_{\cdot}^{\infty}\frac{1}{\sqrt{2 \pi}}e^{\frac{-t^2}{2}} \mathrm{d}t$ is the Gaussian Q-function, and $Q^{-1} (\cdot)$ represents its inverse, $\rho$ is the average SNR which frankly represents the transmit power in watts since the noise is assumed to be normalized to unity. and $|h|^2$ is the envelope of the channel coefficients. The fading coefficients are presented by the random variable $Z=|h|^2$, which is gamma distributed with probability density function given as \cite[Eq. (2.21)]{alouini},\cite{paper9}
\begin{align}\label{pdf}
f_Z(z)=\frac{m^m z^{m-1}}{\Gamma (m)}e^{-m z}, 
\end{align}
where low values of $m$ mark severe fading, high values of $m$ mark the presence of line of sight (LOS) and $m=1$ represents Rayleigh fading. 

\subsection{The relation between Effective Capacity and Effective Energy Efficiency} \label{EC_EEE} 
For low latency communication, effective capacity ($C_e$) is a powerful metric that characterizes the relation between the communication rate and the tail distribution of the packet delay violation probability \cite{mehdi}. For relatively large delay of multiple symbol periods, packet delay violation occurs when a packet delay exceeds a maximum delay bound $\delta$ and the outage probability is defined as \cite{paper6}
\begin{align}\label{delay}
\Lambda=\textrm{Pr}(delay \geq \delta) \approx e^{-\theta \cdot C_e \cdot \delta},
\end{align}	
where $\Pr(\cdot)$ denotes the probability of a certain event. Conventionally, the tolerance of a system to long delays is measured by the delay exponent $\theta$. The system tolerates large delays for small values of $\theta$ (i.e., $\theta\rightarrow 0$), and it becomes stricter delay-wise for large values of $\theta$. As an exemplary scenario, consider 5G NR numerology 1 with symbol period of 35.7 micro-seconds, effective capacity of 1 bpcu and $\theta=0.01$. For a delay outage probability of $\Lambda=10^{-5}$ (i.e, $99.999\%$ reliability, the network can tolerate a maximum delay of $\delta=1151$ symbol periods ($\approx 41$ ms) for $\theta=0.01$, and $\delta=115$ symbol periods ($\approx 4.1$ ms) when $\theta=0.1$.
From \cite{paper5}, the EC in bpcu is 
\begin{align}\label{EC}
C_e(\rho,\theta,\epsilon)=-\frac{\ln\ \psi(\rho,\theta,\epsilon)}{n\theta}, 
\end{align} 
where \vspace{-3mm}
\begin{align}\label{psi} \psi(\rho,\theta,\epsilon)= \mathop{\mathbb{E}_{Z}}\left[\epsilon+(1-\epsilon)e^{-n\theta r(\rho)}\right],
\end{align}
and $\ln$ is the natural logarithm, $r(\rho)$ comes from \eqref{eq3} and $\mathop{\mathbb{E}_{Z}}[\cdot]$ denotes expectation over the fading. The above equations assumes an underlying simple ARQ process and indicate that higher service rates reduce the amount of data bits stored in the queue, and hence also reduce the delay required to transmit, which boosts the EC. Thus, the EC is a measure of the throughput while statistically guaranteeing the delay tolerance. 

\begin{remark}
Note that, we resort to the more practical concept of service rate as in \cite{paper5} rather than the departure rate model in \cite{paper6}. The intuition is that the service rate metric is more suitable for the characterization of the re-invented effective capacity (or effective rate) for discrete short packets operating in the finite blocklength regime. Herein, we consider short packets with a packet drop probability of $\epsilon$ to measure reliability. Reliability could be well mapped via service rate rather than departure rate which does not account for the dropped packets with probability $\epsilon$.  This model allows us to combine the latency and reliability aspects and accommodate them in the characterization of QoS constrained energy efficiency. 
\end{remark}

In \cite{ paper5}, the effective capacity is studied for single node scenario, but never to a closed form expression. It has been proven that the EC is concave in $\epsilon$, and hence has a unique global optimum. From \eqref{eq3} and \eqref{EC}, we observe that increasing the transmission power would definitely raise the EC. However, this comes at the expense of increased power consumption, which is not suitable for energy-limited (battery-operated) IoT devices. Thus, it is necessary to study the trade-off between enhancing EC and power consumption. 

In this context, effective energy efficiency is defined as the achieved effective capacity per unit of consumed power. The EEE captures the trade-off between the throughput of the communication link, the overall power consumption, and latency. Thus, EEE is a suitable metric to quantify and optimize the throughput of the communication link per each consumed watt for energy-limited, low latency IoT. Hence, the optimization of EEE is of great importance for IoT devices, which are isolated from stationary power sources and are required to deliver packets with low latency in the order of milli-seconds \cite{Petreska2016}. In what follows, we study the EEE for short packet communication.

\section{Effective Energy Efficiency under linear power consumption model} \label{EEE_finite blocklength}
In our analysis, we resort to a linear power consumption model 
defined as \cite{Helmy}
\begin{equation}\label{Pt}
P_t(\rho)=\zeta \rho+P_c ,
\end{equation}
with $\zeta\geq 1$ being the inverse drain efficiency of the transmit amplifier and $P_c$ the hardware power dissipated in circuit in watts. 
The linear power consumption model is a well-established and accepted model that has been widely used in many studies related to energy efficiency of wireless systems such as \cite{EEEbook,paper7,Helmy,schober}. This model captures the linear increase of the power consumption as a function of the transmit power and the inverse drain efficiency as well as the idle circuit power consumption. Furthermore, it facilitates the analysis that aims at providing a performance benchmark for the EEE in the context of short packet and low latency communication of IoT devices. 

For this model, the EEE is given by
\begin{equation}\label{EEE0}
\eta_{ee}=\frac{-\frac{1}{n\theta} \ln\left(\mathop{\mathbb{E}_{Z}}\left[\epsilon+(1-\epsilon)e^{-n\theta r}\right]\right) }{\zeta \rho+P_c}. 
\end{equation} 
This scenario assumes an always full buffer and does not account for EBP. In \cite{paper5}, a stochastic model for EC was studied, but never to a closed form expression. Herein, we present a tight approximation for the EC and hence, the EEE.
\begin{lemma} \label{lemma 1}
	The effective capacity in Nakagami-$m$ quasi-static fading is approximated by
	\begin{align}\label{general}
	\begin{split}
	C_e(\rho,\theta,\epsilon)&\approx-\frac{1}{n\theta} \ln\left[\epsilon +(1-\epsilon)\frac{m^m}{\Gamma (m)} \right.\\
	 &\cdot \left. \sum_{n=0}^2\frac{\beta^n}{n!}\int_{0}^{\infty}
	(1+\rho z)^{\alpha} \gamma^n z^{m-1} e^{-mz} dz \right],
	\end{split}
	\end{align} 
	where $\alpha=\frac{-\theta n}{\ln 2}$, $\beta=\theta \sqrt{n} Q^{-1}(\epsilon)\log_2e$, and $\gamma=\sqrt{1-\frac{1}{(1+\rho z)^{2}}}$.
\end{lemma}
\begin{proof}
Please refer to Appendix A
\end{proof}
\begin{remark}
It is quite straightforward to conclude that EEE is an increasing function of the fading parameter $m$. That is, the EEE becomes worse when the fading becomes more severe (i.e, $m\rightarrow \tfrac{1}{2}$). Hence, starting from here, we focus our analysis on quasi-static Rayleigh fading to provide a benchmark of the EC and EEE in the proposed scenarios, where the results can also be extended to any type of Nakagami-$m$ fading. However, Lemma \ref{lemma 1} facilitates the following derivation of the EEE in Rayleigh fading and later, comparing the optimum power allocation in each fading scenario.
\end{remark}
%
\begin{theorem}
\label{lemma 2}
	For a Rayleigh quasi-static fading channel with blocklength $n$, the EEE of the linear power consumption model is approximated as
	\begin{align}\label{Rayleigh}
	\begin{split}
	\eta_{ee}(\rho,\theta,\epsilon)\approx&-\frac{\ln \left[ \epsilon+(1-\epsilon) \ \mathcal{J}\right]}{n \theta \left( \zeta \rho+P_c\right) },  
	\end{split}
	\end{align}	
	where 
	\begin{align}\label{c2.2}
	\mathcal{J}&\!\!=\!e^{\frac{1}{\rho}}\! \rho^\alpha  \!\left[\!\vphantom{\!\frac{\!\Gamma\!\left(\!\alpha\!-\!1,\!\frac{1}{\rho}\!\right) }{\rho^{2}}}\! \left( \!\frac{\beta^2}{2}\!+\!\beta\!+\!1\right)\!  \Gamma\!\left(\!\alpha\!+\!1,\frac{1}{\rho}\! \right) 
	\!-\!\left(\!\frac{\beta^2}{2}\!+\!\beta\right)\!\frac{\!\Gamma\left(\alpha\!-\!1,\frac{1}{\rho}\right) }{\rho^{2}}\!\right]\!,  
	\end{align} 		
	where $\Gamma \left(\cdot,\cdot\right)$ is the upper incomplete gamma function \cite{Abramowitz}. 
\end{theorem}
\begin{proof} Please refer to Appendix B.
\end{proof}
\begin{remark} \label{lemma 3}
It was shown in \cite{iswcs_eee} that both EC and EEE are concave functions of the error probability $\epsilon$, and the optimum value of $\epsilon$ that maximizes them in this case is given by 
 	\begin{align}\label{e*}
 	\begin{split}
 	\epsilon^*(\rho,\alpha,\beta)\approx\arg\min_{0 \leq \epsilon \leq 1} \  \epsilon+(1-\epsilon) \ \mathcal{J}.
 	\end{split}
 	\end{align} 
\end{remark}
In what follows, we study the behaviour of EEE as a function of transmit SNR $\rho$.

\begin{theorem} \label{lemma quasi-concavity}
The EEE function in Theorem \ref{lemma 2} is a quasi-concave function of the transmit SNR.
\end{theorem}
\begin{proof} 
Please refer to Appendix C
\end{proof}
\begin{figure}[!b] 
	\centering
	\includegraphics[width=1\columnwidth]{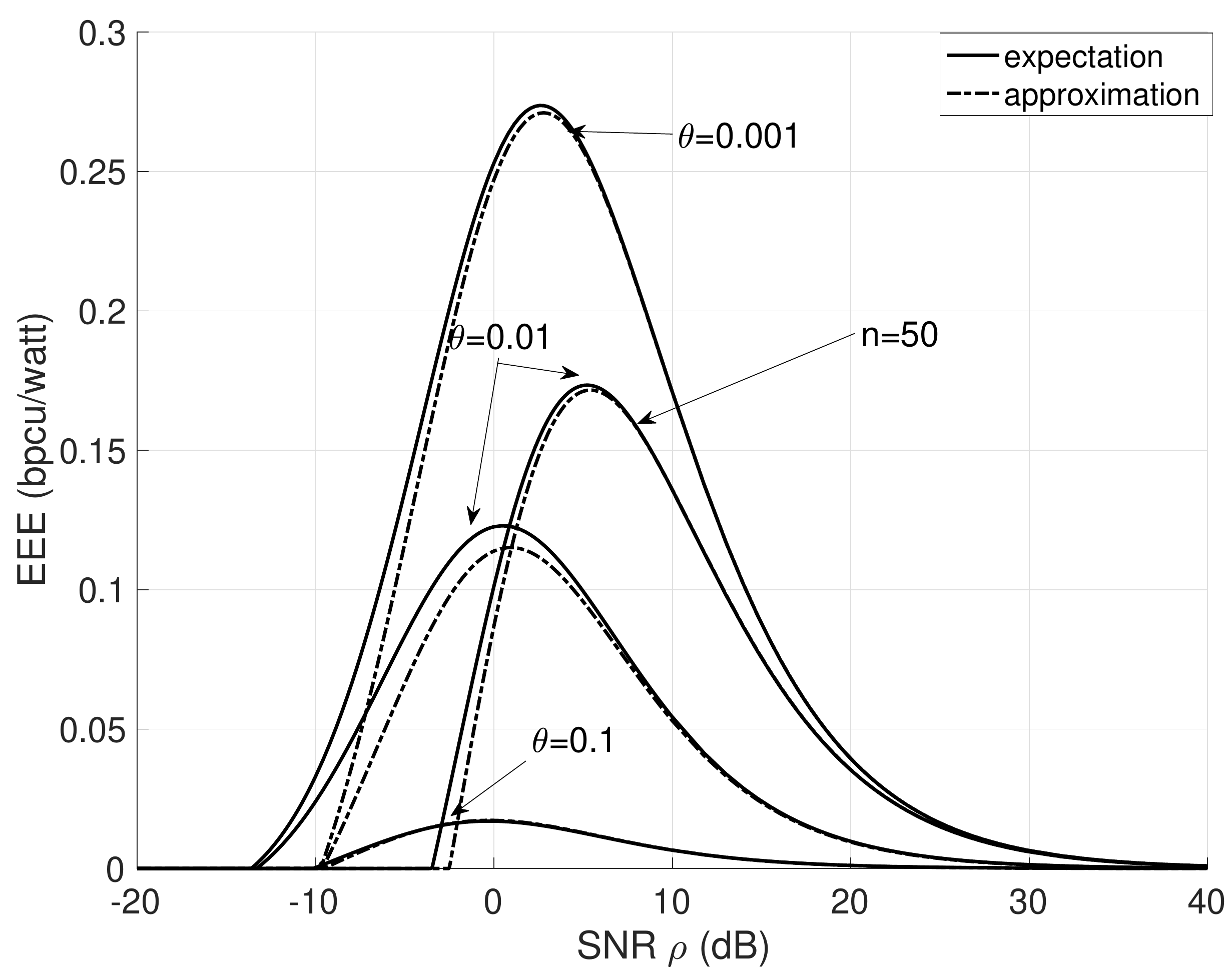}
	\centering
	\vspace{-5mm}
	\caption{EEE vs SNR in quasi-static fading for $m=1$, $n=500, \epsilon=10^{-4}$, $P_c=1.2, \zeta=1.2, \lambda=1$ and different delay exponents $\theta$.}
	\vspace{-5mm}
	\label{EEE}
\end{figure}
Fig. \ref{EEE} illustrates the EEE in Rayleigh block fading channel for different delay exponents, while applying the expectation in (\ref{EEE0}) and Theorem \ref{lemma 2}. The network parameters are $n=500$ symbol periods and $\epsilon=10^{-4}$. The figure proves the accuracy of Theorem \ref{lemma 2} as a tight approximation for the EEE, which is also well established for the EC in \cite{TVT_Shehab}. However and unlike the EC which is an asymptotically increasing function of the SNR, the figure shows the convexity of the upper contour of the EEE in the transmit power and that the approximation in Theorem \ref{lemma 2} captures this quasi-concavity precisely as stated in Theorem \ref{lemma quasi-concavity}. Note that the EEE declines when the delay constraint becomes more strict. Meanwhile, it is also observed that the optimum transmit power shifts to a higher value for less strict delay constraints. Finally, we plot the EEE also for smaller packets with length of $n=50$ symbol periods. As we can observe, the EEE is higher for smaller packet size as the delay is minimized which boosts the EC. The approximation holds tightly for this setup as well.

The quasi-concavity of the EEE in the SNR aids to characterize the optimum power allocation which maximizes the EEE in the linear power consumption model. 
\begin{theorem}\label{th-p*}
The optimum power allocation for maximizing the EEE is $\rho^*$, which is the solution of 
\begin{align}\label{p*9}
\begin{split}
\eta_{ee}(\rho^*)=-\frac{1}{n \theta}\left( \frac{\mathcal{J}^{'}(\rho^*)}{\mathcal{J}(\rho^*)}\right), 
\end{split}
\end{align}
where $\kappa_1= \frac{\beta^2}{2}+\beta+1$ and $\kappa_2=\frac{\beta^2}{2}+\beta$, and
\begin{align}\label{p*10}
\begin{split}
\mathcal{J}=e^{\frac{1}{\rho}} \rho^{\alpha}\left(\kappa_1 \Gamma\left( \alpha+1,\frac{1}{\rho}\right) -\frac{\kappa_2}{\rho^2} \Gamma\left( \alpha-1,\frac{1}{\rho}\right)  \right),
\end{split}
\end{align}
\begin{align}\label{p*11}
\begin{split}
&\mathcal{J}^{'}(\rho^*)=\frac{\!\partial \mathcal{J}}{\!\partial \rho} \\
&=-\frac{1}{\rho^2}\left[\left(1+\frac{\alpha}{\rho} \!\right)\!\mathcal{J} +\!\frac{(1-\kappa_1)e^{-\frac{1}{\rho}}}{\rho^{\alpha}}- \frac{2 \kappa_2}{\rho}  \Gamma\left( \alpha\!-\!1,\frac{1}{\rho}\right) \right]\!. 
\end{split}
\end{align}
\end{theorem}
\begin{proof}
Based on the quasi-concavity of the EEE function which was proven in Theorem \ref{lemma quasi-concavity}, we differentiate (\ref{Rayleigh}) with respect to $\rho$ and equate to zero as follows
\begin{align}\label{p*7}
\begin{split}
\frac{\partial \eta_{ee}}{\partial \rho}= -\left[\frac{\frac{(1-\epsilon)\mathcal{J}^{'}(\zeta \rho+P_c)}{\epsilon+(1-\epsilon)\mathcal{J}}-\zeta \ln(\epsilon+(1-\epsilon)\mathcal{J})}{n \theta(\zeta\rho+P_c)^2} \right] \\ 
\approx-\left[\frac{\frac{\mathcal{J}^{'}(\zeta \rho+P_c)}{n \theta\mathcal{J}(\zeta\rho+P_c)}-\frac{\zeta \ln(\epsilon+(1-\epsilon)\mathcal{J})}{n \theta(\zeta\rho+P_c)}}{(\zeta\rho+P_c)} \right]=0,
\end{split}
\end{align}
where the above approximation is valid since ${\mathcal J}$ is much larger than 1, and the reliability constraint $\epsilon$ is very small (i.e, $ {\mathcal J}>>\epsilon$).  Then, to differentiate $\mathcal{J}$, we apply the derivative of the upper incomplete gamma function \cite{Gradshteyn}, which yields 
\begin{align}
\frac{\!\partial\! \mathcal{J}}{\!\partial\! \rho}&\!=\!-\frac{1}{\rho^2}\!\left[\!\mathcal{J}\!+\! \frac{\alpha}{\rho} \mathcal{J}\!-\! \frac{\kappa_1 e^{\!-\frac{1}{\rho}}}{\rho^\alpha} \right. 
\left.- \frac{2 \kappa_2}{\rho}  \Gamma\!\left(\! \alpha\!-\!1,\frac{1}{\rho}\right)\!+\!\frac{e^{\!-\frac{1}{\rho}}}{\rho^\alpha}  \!\right], 
\end{align}
and after algebraic manipulation we obtain \eqref{p*11}, which concludes the proof. 
\end{proof}
Despite the fact that we were able to find the partial derivative of $\mathcal{J}$, a closed form solution for (\ref{p*9}) does not exist. For this purpose, we can compute a point-wise numerical solution or utilize Matlab root-finding functions, e.g., fzero in a similar way to \cite{paper8}.

\section{Empty buffer probability model} \label{EBP}
Previously, we assumed that the buffer is always full which practically is not always the case. In real scenarios, there would be instants in which a certain IoT device becomes idle and therefore has no data to transmit. Thus, we need to account for the case when the buffer is empty. Accordingly, we apply the model considered in \cite{paper7} to networks operating in the finite blocklength regime with non-vanishing probability of error $\epsilon$. After accounting for EBP, the transmission probability $P_{nb}$ is equal to ($1-$ the probability of empty buffer) and the transmission process appears in Fig. \ref{Empty_buffer}. 
\begin{figure}[!t] 
	\centering
	\includegraphics[width=1\columnwidth]{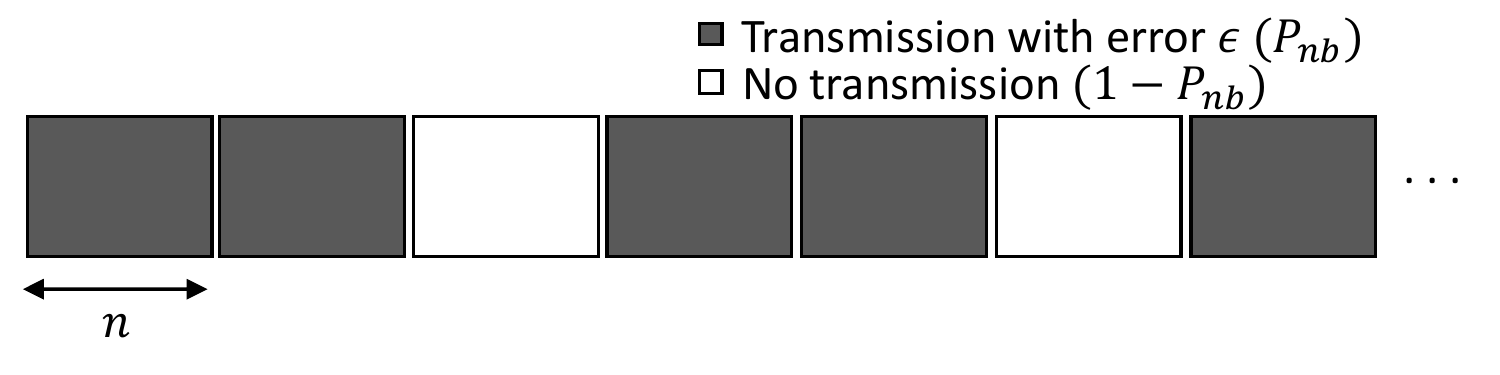}
	\vspace{-4mm}
	\caption{Transmission with empty buffer probability in quasi-static channel with blocklength $n$.}
	\label{Empty_buffer}
	\vspace{-5mm}
\end{figure}

For an average arrival rate of $\lambda$ and a stable queue, the power consumption becomes 
\begin{equation}\label{Pt_nb}
P_t(\rho)=P_{nb}\zeta p+P_c=\frac{\lambda}{\mathop{\mathbb{E}}\left[r\right] }\zeta \rho+P_c,
\end{equation}
with $P_{nb}=\frac{\lambda}{\mathop{\mathbb{E}}\left[r\right] }$ denoting non-empty buffer probability (NBP), which is bounded between 0 and 1. The EEE with EBP is 
\begin{equation}\label{EEE1}
\eta_{ee}=\frac{-\frac{1}{n\theta}\ln \left[ \epsilon+(1-\epsilon) \ \mathcal{J}\right]}{ \frac{\lambda}{\mathop{\mathbb{E}}\left[ r\right] }\zeta \rho+P_c},  
\end{equation} 
where the numerator represents the effective capacity in the finite blocklength regime as defined in Theorem \ref{lemma 2}.

Note that similar to the fluid model\footnote{The reason behind the choice a fluid model is that these fluid models are motivated as approximations to discrete queueing systems. Fluid flow queues have been well accepted as a useful mathematical tool for modeling and have long been used to evaluate the performance of telecommunication and computer systems. In particular, we apply the fluid model to characterize the asymptotic delay probability, and to approximate the NBP by arrival rate divided by average service rate. The exact (non-asymptotic) analysis of delay outages and empty buffers is outside the scope of this work.}, the NBP indicates the average asymptotic probability of transmission over relatively long time, where $P_{nb}=1$ means that the average serviced amount of data per packet time is equal to the average amount of data arrival per packet time. This indicates that, in average, transmission always occurs. 

\subsection{Verifying the effective energy efficiency model with empty buffer probability in finite blocklength}

According to \cite{EEEbook}, an energy efficiency function must be non-negative, must be zero when the transmit power is zero, and must tend to zero as the transmit power tends to infinity. It was shown in \cite{paper7} that this power the EBP model fulfills is valid for Shannon model. In the following Lemma, we verify that this EBP power consumption model is valid as well for short packets transmission. 
\begin{lemma} \label{EEEV}
	The EEE in (\ref{EEE1}) is zero for $\rho=0$ and tends to 0 when $\rho\rightarrow \infty$.
\end{lemma}
\begin{proof} Please refer to Appendix D.
\end{proof}

\subsection{Effective energy efficiency maximization with buffer constraints}
We investigate the EEE maximization with EC, delay, and power constraints. EC should be larger than the arrival rate $\lambda$ to guarantee a stable queue, while the transmission SNR $\rho$ is bounded by $\rho_{max}$. Thus, the optimization problem is formulated as  
\begin{equation}\label{op2}
\begin{split}
\max_{\rho \geq 0, \theta \geq 0} \ &\eta_{ee}=\frac{-\frac{1}{n\theta}\ln \left[ \epsilon+(1-\epsilon) \ \mathcal{J}\right]}{P_{nb}\zeta \rho+P_c},  \\ 
s.t \ \ &C_e(\rho,\theta,\epsilon)\geq\lambda, \ \ P_{nb} e^{-\theta \lambda \delta}\leq \Lambda \\
&\rho \leq \rho_{max}, \ \ \epsilon \leq \epsilon_t, \ \ 0 \leq P_{nb} \leq 1.
\end{split}
\end{equation} 
Note that the reliability constraint here, which is the error probability constraint on the first transmission, is important to improve QoS in URLLC. Meanwhile, the delay outage probability constraint does not necessarily guarantee reliability if the $\epsilon$ is not imposed.

For the full buffer model, we set $P_{nb}$ to 1. We perform a line search for $\rho$ in the interval $\left[0,\rho_{max} \right]$. The optimum error probability is $\min\left[\epsilon^*,\epsilon_t \right]$ where $\epsilon^*$ is obtained from Remark \ref{lemma 3}. When analyzing the empty buffer scenario, we set $P_{nb}=\frac{\lambda}{\mathop{\mathbb{E}}\left[ r\right] }$. Here, $\Lambda$ is the maximum allowed delay outage probability. In all cases, the optimal value of $\theta$ can be obtained from the second constraint at equality as
\begin{equation}\label{th_pb}
\theta^{*}(\rho)=\frac{1}{\lambda \delta} \ln \frac{P_{nb}}{\Lambda }.
\end{equation}
\section{Adaptive Retransmission scenario} \label{retransmission}
As shown in \cite{Dosti}, achieving ultra-reliability using one transmission consumes a huge amount of power, which is not applicable for energy limited IoT devices. Herein, we present a basic framework for applying ARQ by considering adaptive rate retransmission of faulty packets when the buffer is empty in the EBP model in order to achieve ultra-reliability. In this framework, at a certain time instant $t$, we assume a buffer-aware transmission as in \cite{buffer1,buffer2}; this allows the transmitter to have a prior knowledge of whether the buffer will be empty or there is a packet that needs to be delivered at the next time slot $t+1$. Then the following is applied:
\begin{enumerate}
	\item If a packet arrives at time slot $t$ and there is also a packet arrival at $t+1$, normal transmission occurs at $t$ with rate $r(\epsilon)$.
	
	\item If a packet arrives at time slot $t$ and there is no packet arrival at $t+1$, we transmit the with rate $r(\epsilon_1>\epsilon)$ and apply ARQ at $t+1$.  
	
\end{enumerate}

However, the non-EBP in this case will be slightly changed to $P_{nb}^{'}$ due to the rate variations. The first case occurs when the buffer is full at $t+1$. Assuming i.i.d arrivals, the probability of occurrence of case number 1 is the same as the non-empty buffer probability $P_{nb}^{'}$.

In case 2, we apply the type-I ARQ protocol. In type-I ARQ protocol, the node is allowed to retransmit its packet if it receives a NACK feedback from the receiver, which indicates that the packet is not successfully decoded. We assume a maximum of only 1 retransmission in order to satisfy the stringent delay requirements in URLLC. Note that, the transmitter performs the first transmission with an error probability of $\epsilon_1$ and the second transmission with an error probability of $\epsilon_2$ such that the aggregate error probability satisfies the reliability constraint (i.e, $\epsilon_1 \epsilon_2 \leq \epsilon$). Thus, both $\epsilon_1$ and $\epsilon_2$ are higher than $\epsilon$. The fact that the transmission rate is an increasing function of the error probability implies that both $r(\epsilon_1)$ and $r(\epsilon_2)$ are higher than $r(\epsilon)$, which reflects the rate gains of this model. In other words, we obtain a significant rate gain in the likely event of successful first transmission. Moreover, the aggregate packet error probability when applying ARQ would be $\epsilon$ which maintains the reliability level. This can be performed by rate adaption as in \cite{r_adaption}, where sensors generate data in the form of bits that can be relocated from one packet to the other. This rate adjustment occurs when the buffer is empty in the next time instant $t+1$ with probability $1-P_{nb}^{'}$ and results in a rise in the instantaneous rate with symbols resizing. Generally let $r_0=\mathbb{E}\left[r(\epsilon)\right]$, $r_1=\mathbb{E}\left[r(\epsilon_1)\right]$ and $r_2=\mathbb{E}\left[r(\epsilon_2)\right]$. Then the average rate becomes
\begin{align} \label{e44}
\mathbb{E}\left[r\right]&=p_{nb}^{'}r_0+(1-p_{nb}^{'})\left[(1-\epsilon_1)r_1+\frac{\epsilon_1r_2}{2} \right] \notag \\
&=p_{nb}^{'}\left(r_0-\kappa \right)+\kappa, 
\end{align}
where $\kappa=(1-\epsilon_1)r_1+\frac{\epsilon_1r_2}{2}$. Note that the rate is divided by 2 in case of retransmission because the time duration of transmitting $n$ symbols is approximately the double as expressed in the last term of (\ref{e44}). Due to the rate variation, the modified non-EBP should satisfy
\begin{align} \label{e66}
p_{nb}^{'}=\frac{\lambda}{\mathbb{E}\left[ r\right] }=\frac{\lambda}{p_{nb}^{'}\left(r_0-\kappa \right)+\kappa}.
\end{align}

Solving (\ref{e66}) for $p_{nb}^{'}$, we obtain
\begin{align} \label{e78}
p_{nb}^{'}=\frac{-\kappa+\sqrt{\kappa^2+4(r_o-\kappa)\lambda}}{2(r_o-\kappa)} \leq 1
,
\end{align}
and the EC in this case is given by
\begin{align}\label{e144}
&C_{e2}=\frac{-1}{n\theta}\ln\!\left( \mathop{\mathbb{E}_{Z}}\left[p_{nb}^{'}(\epsilon+(1-\epsilon)e^{-n\theta r_0})+ \right.\right. \notag \\
&\left.\left.(1-p_{nb}^{'})\left( (1-\epsilon_1)e^{-n\theta r_1}+\epsilon_1(1-\epsilon_2) e^{-n\theta \frac{r_2}{2}}+\epsilon\right) \! \right] \right). 
\end{align}

To clarify \eqref{e144}, we mention that the transmission occurs with rate $r_0$ when the buffer is not empty in the next time slot. This is indicated by the first term in equation (25) and occurs with probability $p_{nb}^{'}$. When the buffer is empty in the next slot with probability $(1-p_{nb}^{'})$, one transmission occurs rate $r_1$ if there is no error where the no error probability is $1-\epsilon_1$. However, the rate is divided by 2 to become $r_2/2$ in case of retransmission when an error occurs in the first transmission only with probability $\epsilon_1 (1-\epsilon_2)$. While the rate is considered to be zero when both transmissions fail with probability $\epsilon_1 \epsilon_2=\epsilon$. In order to map the impact of the second transmission on the overall effective capacity, we calculate the average of the rate of the first transmission which is zero in case of error and the second transmission which is $r(\epsilon_2)$. Since the transmission occurs in the duration of 2 time slots, the rate is virtually divided by 2 as indicated in the second term of \eqref{e144}.

Since the channel coefficients change from one transmission to the other, we need to vary the transmit power in order to compensate for the channel coefficient variation so that the product $\rho |h|^2$ is the same for both transmissions. Although the transmit power varies, the average transmit power $\mathbb{E}[\rho]$ is the same for both transmissions and independent from $p_{nb}^{'}$ or any other parameter. Thus, the consumed power for this scenario is a probabilistic function of the transmit power of one or two transmissions which after manipulations is allowed to be expressed as \vspace{-2mm}
\begin{align} \label{e16}
P_t=\left[ p_{nb}^{'2}+p_{nb}^{'}(1-p_{nb}^{'})  \left( 1+\epsilon_1\right)\right]\zeta \rho + P_c ,
\end{align}
Defined by the quotient of EC to the transmit power, the EEE of this scenario is formulated as
\begin{align}\label{e18}
\eta_{ee_2}=\frac{C_{e2}}{\left[  p_{nb}^{'2}+p_{nb}^{'}(1-p_{nb'}^{'})  \left( 1+\epsilon_1\right)\right]\zeta \rho + P_c}. 
\end{align}

Note that the optimum error probability of each transmission in the EBP model with two ARQ transmissions is not simply the square root of the aggregate target error probability $\epsilon=\epsilon_t$. Therefore, we define the optimization problem to determine the optimum error probability of the first transmission that maximizes the EEE subject to a target reliability constraint as
\begin{align}\label{e82}
\max \ \ &\eta_{ee_2}(\epsilon_1,\epsilon_2) \\
s.t \ \ &0<\epsilon_1 \epsilon_2 \leq \epsilon_t \notag
\end{align}
An interesting analysis is to determine the asymptotic behaviour of the EEE as the delay constraint $\theta \rightarrow$ $\infty$ or 0. This corresponds to extremely strict or no delay constraint, respectively. It is straight forward to conclude that the EEE tends to zero for extremely stringent delay constraint (i.e, when $\theta \rightarrow$ $\infty$). Herein, we derive the upper bound of the EEE for relaxed latency constraint as $\theta \rightarrow 0$. 

\begin{theorem} \label{ub}
The EEE of the proposed retransmission scenario is upper bounded by 
\begin{align} \label{ub eq}
\!\!\!\lim\limits_{\theta\rightarrow 0} &\eta_{ee_2\!\!}\!=\!\! 
\frac{p_{nb}^{'}(1\!-\!\epsilon)r_o\!+\!(1\!-\!p_{nb}^{'})\!\left[\!(1\!-\!\epsilon_1) r_1\!+\!\epsilon_1\!(1\!-\!\epsilon_2\!)\!\frac{r_2}{2} \!\right]}{\!\left[\!  p_{nb}^{'2}\!+\!p_{nb}^{'}(1\!-\!p_{nb'}^{'})  \!\left(\! 1\!+\!\epsilon_1\!\right)\!\right] \zeta \rho \!+\! P_c}\!,
\end{align}	
and lower bounded by zero.
\end{theorem}
\begin{proof}
	Please refer to Appendix E.
\end{proof}
\begin{remark}
As the system approaches ultra-reliability (i.e, $\epsilon\rightarrow 0$ ), the EC of one transmission converges to $r_0$, while the EC of the proposed EBP-ARQ scehme converges to $p_{nb}^{'}r_o+(1-p_{nb}^{'})r_1$. Hence, the EC is raised to $r_1$ for $(1-p_{nb}^{'})$ portion of the time which indicates the gain in the EC of the proposed EBP retransmission scheme. 
\end{remark}

\subsection{Power saving}
Returning to \eqref{e16}, which represents the average power consumption, we study the effect of varying the non-EBP $p_{nb}^{'}$ on the power consumption by obtaining the first derivative of \eqref{e16} as 
\begin{align} \label{e10}
\frac{\partial P_t}{\partial p_{nb}^{'}}=\left[ -2 \epsilon_1 p_{nb}^{'}+(1+\epsilon_1)\right] \zeta \rho ,
\end{align}
which is strictly non-negative for all possible values of $p_{nb}^{'}$ and $\epsilon_1$ (i.e, $0\leq p_{nb}^{'}, \epsilon_1 \leq 1$). Thus, the power consumption $p_t$ is still an increasing function of the non-empty buffer probability $p_{nb}^{'}$. Hence, minimizing the non-empty buffer probability also reduces the transmit power which leads to a longer battery life for remote sensors that are located far from energy sources.

 \begin{theorem} \label{lemma pc}
The non-EBP $p_{nb}^{'}$ in \eqref{e78} is a pseudo-convex function in $\epsilon_1$ and therefore, the minimization of $p_{nb}^{'}$ is a fractional program.
\end{theorem}
\begin{proof}
Please refer to Appendix F.
\end{proof}

Furthermore, being a psuedo convex fractional program and due to the analytical intractability, it is easier to find the global minimum of $p_{nb}^{'}$ using well known optimization algorithms such as Dinkelbach's algorithm \cite{EEEbook} to minimize $p_{nb}^{'}$ and hence, the total transmit power. Later the results section shows that this optimal solution for minimizing the transmit power also highly approaches optimality for maximizing the EEE in this case. Moreover, it is more efficient and numerically tractable\footnote{Taking into account the energy consumption and computational complexity for the original online optimization algorithm, it might even turn out that the suboptimal one proposed here leads to an overall better effective energy efficiency and lower latency.} to minimize the transmit power instead of the EEE function given in \eqref{e82}. Hence, the problem becomes \vspace{-0mm}
\begin{align}\label{e84}
\min \ \ &p_{nb}^{'}(\epsilon_1,\epsilon_2) \\
s.t \ \ &0<\epsilon_1 \epsilon_2 \leq \epsilon_t. \notag
\end{align}

Note that the target is to minimize the transmit power. Thus, it is straight forward to conclude that the reliability constraint is optimally achieved at equality since more power is needed to achieve lower error and higher reliability.
As proven in Theorem \ref{lemma pc}, the problem in \eqref{e84} is a pseudo-convex fractional program. Therefore, the global optimum exists, and can be found by utilizing the Dinkelbach's algorithm as introduced in Section 3.2 in \cite{EEEbook}. It is a parametric algorithm of which the basic idea is to tackle a pseudo-convex problem by solving a sequence of easier problems which are guaranteed to converge to the global optimum. The minimization procedure is depicted in Algorithm 1.
\begin{algorithm}
	\DontPrintSemicolon
	\SetAlgoLined
	\SetKwInOut{Input}{Input}\SetKwInOut{Output}{Output}
	\Input{$F_{\sigma_0} > \delta > 0$; $n = 0$; $\sigma = 0$;}
	\Output{$\epsilon_1^*$}
	\BlankLine
	
	\While{$F_{\sigma_n} > \delta$}{
		$\epsilon_1^* = \arg \min \{-\kappa(\epsilon_1)+\sqrt{\kappa(\epsilon_1)^2+4(r_o-\kappa(\epsilon_1))\lambda} - \sigma_n2(r_o-\kappa(\epsilon_1))\}$;
		
		$F_{\sigma_n} = -\kappa(\epsilon_1^*)+\sqrt{\kappa(\epsilon_1^*)^2+4(r_o-\kappa(\epsilon_1^*))\lambda} - \sigma_n2(r_o-\kappa(\epsilon_1^*))$;
		
		$\sigma_{n + 1} = \frac{-\kappa(\epsilon_1^*)+\sqrt{\kappa(\epsilon_1^*)^2+4(r_o-\kappa(\epsilon_1^*))\lambda}}{2(r_o-\kappa(\epsilon_1^*))}$;
		
		$n = n + 1$;
	}
	
	\caption{Minimization of $p_{nb}^{'}$}
\end{algorithm}

The intuition behind the algorithm is as follows. It starts from some arbitrary estimate of $\epsilon_1^*$ and analyzes the level sets of the original problem, which are evidently convex. Then, as the algorithm progresses, it iteratively corrects the estimate of $\epsilon_1^*$ and checks if the stopping criterion is satisfied, i.e. a $\delta$-suboptimal solution has been obtained. If the tolerance margin has not yet been satisfied, then the algorithm continues scanning through the level sets of the function until convergence. Notice that the worst-case computational complexity of the algorithm is dominated by step 2, which can be solved using interior point methods. As a consequence, the convergence rate in the sub-problem sequence is super-linear \cite{Boyd,Nesterov}.

\subsection{Average latency}

Herein, we analyze the average extra packet delay induced due to retransmissions when applying the proposed EBP-ARQ with empty buffer instants. Let $\delta_1$ be the delay per packet in the single transmission scenario and $\delta_2$ be the total delay when two transmissions occur. Then the expected delay when applying ARQ with two retransmissions would be
\begin{align} \label{e30.1}
\tau=\delta_1\left( P^{'}_{nb}+(1-P^{'}_{nb})(1-\epsilon_1)\right) +\delta_2(1-P^{'}_{nb})\epsilon_1.
\end{align} 
In case of error in the first transmission, the 1 bit NACK feedback could be transmitted in a span of $\approx 6$ symbols according the Physical Uplink Control Channel (PUCCH) Format 1 in 5G NR \cite{5G_NR}. Assuming the same transmission rate for the NACK packet, the extra delay due to the NACK packet would be $\Delta=\frac{6}{n}\delta_1$. This occurs during before the second transmission and is very small compared to the one packet transmission time. Thus, we can state that $\delta_2=2\delta_1+\Delta=\left(2+\frac{6}{n}\right)\delta_1$. Hence, the normalized delay with respect to one transmission time $\delta_1$ can be written as
\begin{align} \label{e30.4}
\tau_n= P^{'}_{nb}+(1-P^{'}_{nb})(1-\epsilon_1) +\left(2+\frac{6}{n}\right)(1-P^{'}_{nb})\epsilon_1,
\end{align} 
where $\tau_n\geq 1$. Herein, \eqref{e30.4} provides an indication of the QoS when applying ARQ wiht EBP for boosting the EEE. Note that we mainitain the same QoS constraints $\theta$ snd $\epsilon$ throughout the whole analysis.

\section{Results and discussion} \label{results}
In this section, we present numerical results to illustrate the behaviour of the EEE function and the trade off between the EEE, power allocation and latency for Shannon's model and finite blocklength in different transmission scenarios. We compare our results to the infinite blocklength case to show the performance gap that results from applying the short packet information theoretic approach which is more suitable for delay constrained analysis and compare this gap to the long packets ideal case. Firstly, Fig. \ref{f1} illustrates the EEE of short packet transmission as a function of the delay exponent $\theta$ in quasi-static Rayleigh fading for $n=500$, $\rho=3$  dB, error probability $\epsilon=10^{-4}$, and different circuit powers $P_c$. The figure highlights the energy efficiency gap between long packet transmission which is analyzed via Shannon capacity model and the finite blocklength model. The EEE of short packets is less than the infinite blocklength Shannon's model by about 20\% in this case. Moreover, the figure shows that the EEE declines when the delay exponent becomes more strict and when the consumed power in circuitry is higher. 

\begin{figure}[!t] 
	\centering
	\includegraphics[width=1\columnwidth]{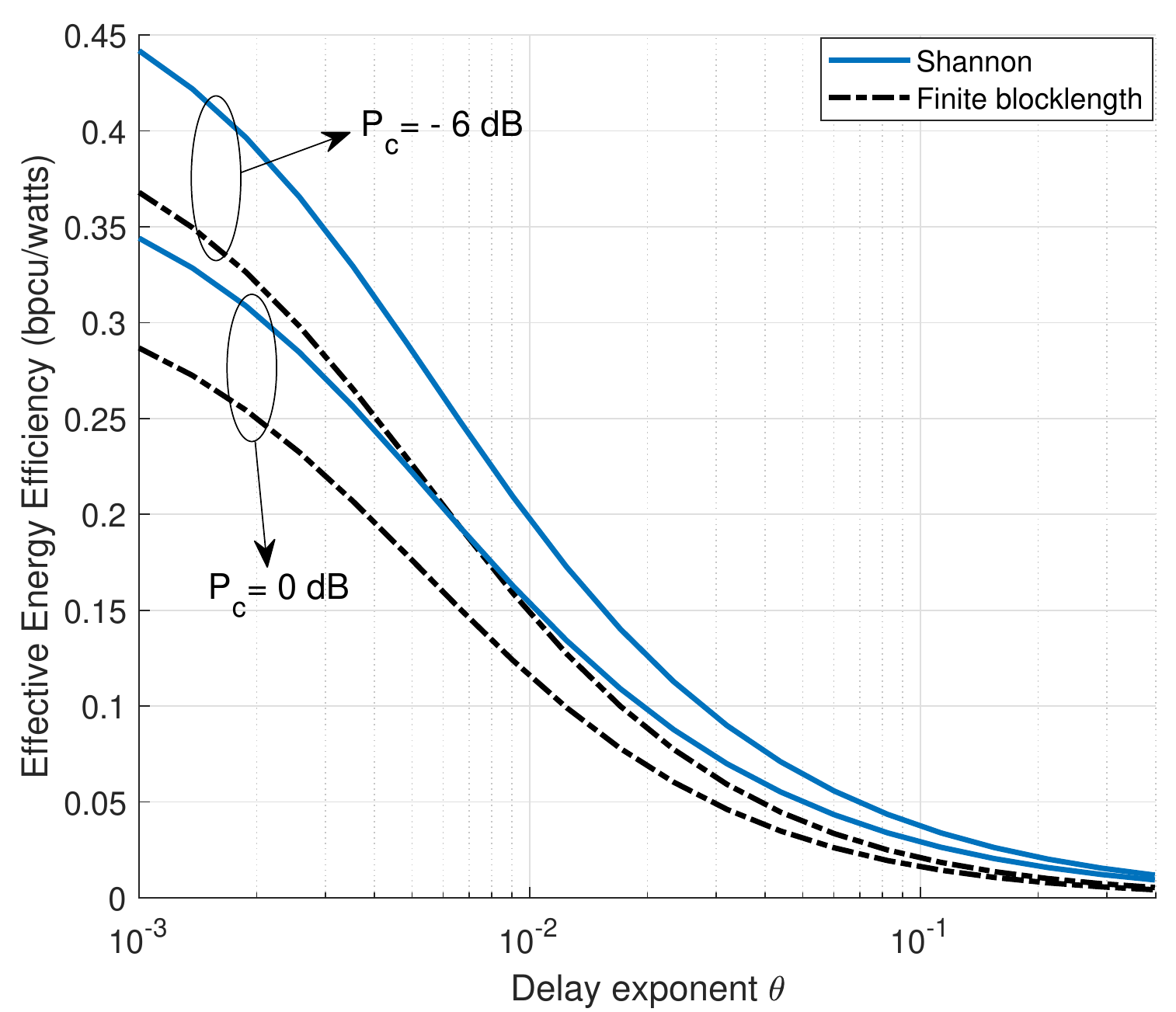}
	\centering
	\vspace{-2mm}
	\caption{Effective energy efficiency as a function of the delay exponent $\theta$ in quasi-static Rayleigh fading for $n=500$, $\rho$=3 dB, error probability $\epsilon=10^{-4}$, and different circuit powers $P_c$.}
	\vspace{-2mm}
	\label{f1}
\end{figure}

In Fig. \ref{f2}, we elucidate the EEE for different transmission probabilities (i.e, when the buffer is not empty). The figure shows the EEE gap between infinite and finite blocklength models. Again, it is noted that higher circuit power significantly deteriorates the EEE. It is observed that the EEE monotonically decreases with the increase of arrival rate (or alternatively Non-EBP) which indicates higher congestion in the network. However, this effect becomes marginal when the circuit power is higher as the circuit power becomes a dominant factor in the calculation of EEE. Thus, for $P_c= 0$ dB, the EEE is nearly constant as a function of the arrival rate. Therefore, careful studying of EEE for different source arrival rates is crucial for low circuit power.

\begin{figure}[!t] 
	\centering
	\includegraphics[width=0.95\columnwidth]{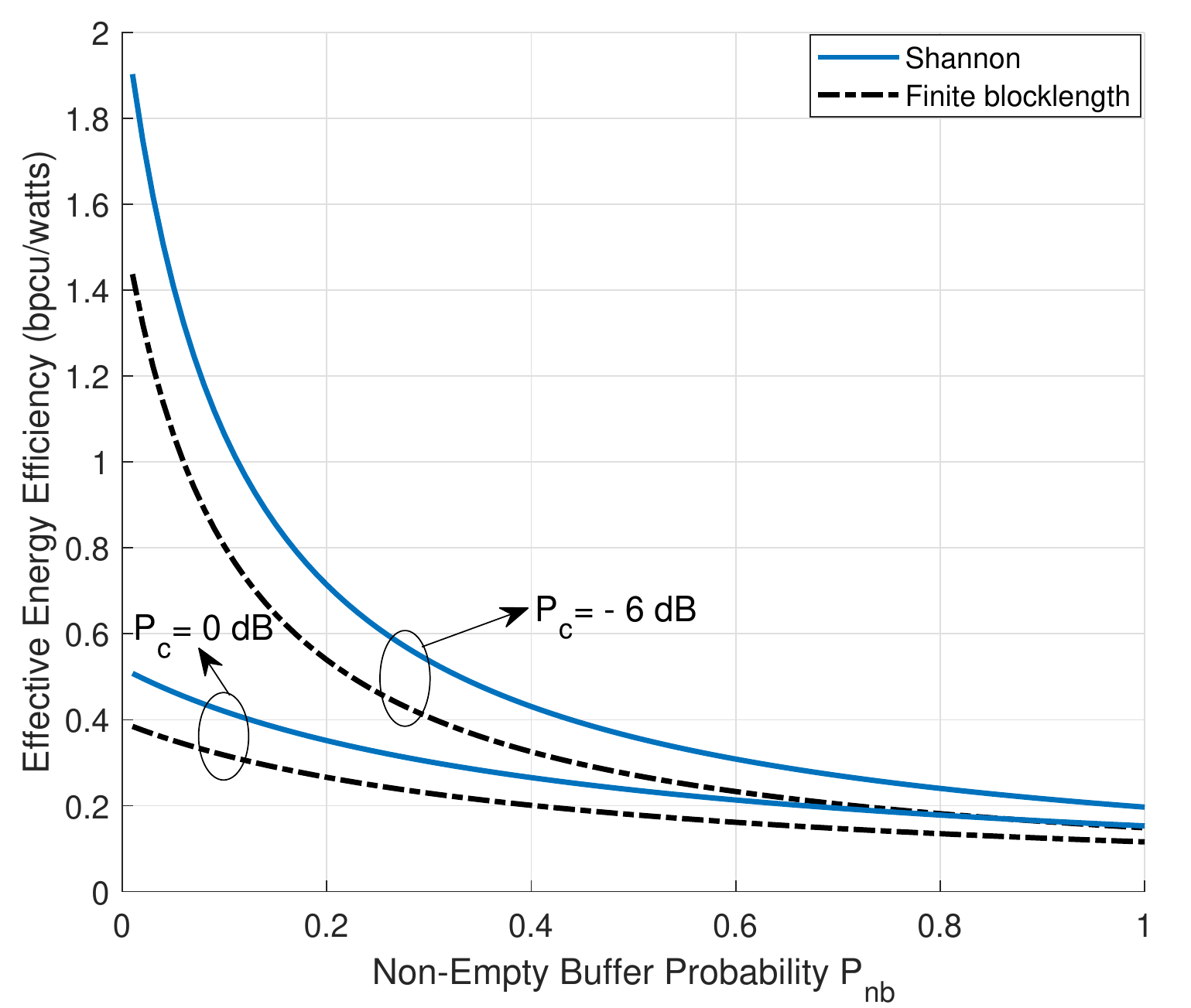}
	\centering
	\vspace{-2mm}
	\caption{Effective energy efficiency vs as a function of the non-EBP $P_{nb}$ in quasi-static Rayleigh fading for $n=500$, SNR=3 dB, error probability $\epsilon=10^{-4}$, $\theta=0.01$ and different circuit powers $P_c$.}
	\vspace{-2mm}
	\label{f2}
\end{figure}

For the following simulations, we fix the network parameters as follows: $\Lambda=\left\lbrace 10^{-2}, 10^{-3}\right\rbrace , P_c=0.2 \ W, \ \zeta=1.2, \ \lambda=1, \ \delta=500$ symbol periods, and $n=500$ symbol periods, unless stated otherwise. In Fig. \ref{EEE_e}, we evaluate the EEE as a function of error probability $\epsilon$ in case of EBP and compare it to the case where the buffer is always full while fixing the transmit power at $\rho=10$ dB. We observe that the EEE is concave in $\epsilon$ as stated in Remark \ref{lemma 3}. It is obvious that considering the probability of empty buffer reflects a gain in the EEE over the full buffer model, while decreasing the delay outage probability reduces the EEE. Moreover, the figure depicts that Shannon's model considered in \cite{paper7} overestimates the EEE by more than $20 \%$ when compared to the finite blocklength model.  
\begin{figure}[!t] 
	\centering
	\includegraphics[width=1\columnwidth]{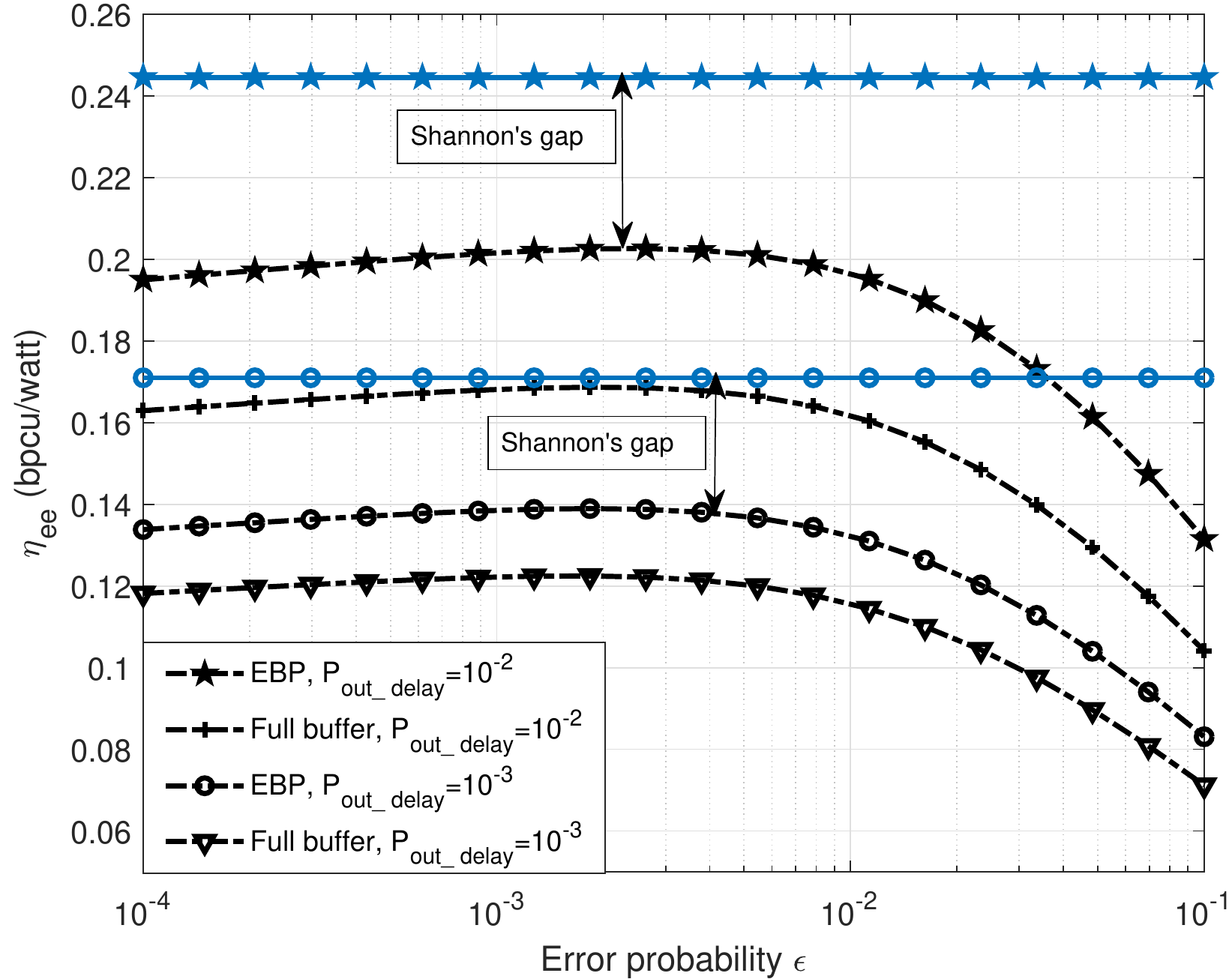}
	\centering
	\vspace{-2mm}
	\caption{EEE vs $\epsilon$ with and without empty-buffer probability for $\Lambda=10^{-2}, 10^{-3}, P_c=0.2, \zeta=1.2, \lambda=1, \delta=500$, and $n=500$.}
	\vspace{-2mm}
	\label{EEE_e}
\end{figure}

Fig. \ref{EEE_Dmax} depicts the maximum achieved EEE obtained from (\ref{op2}) for different delay limits $\delta$, where $\rho_{max}=13$ dB (variable transmission power), and $\epsilon_t=10^{-4}$. We observe that the EEE increases when extending the delay bound $\delta$ and relaxing the delay outage probability $\Lambda$. This implies that networks which can tolerate longer packet transmission delay are more energy efficient. From another perspective, it is clear that the sporadic non-EBP transmission scenario allows for a better modelling of the power consumption in MTC. This reflects that full buffer is the worst case, where we assume that all power will be consumed. Meanwhile, the Non-EBP models the fraction of time that is actually used for transmission of packets according to the queue congestion, which interprets the gain of this model compared to always full buffer. Furthermore, the figure verifies the inaccuracy of Shannon's model when computing the EEE for relatively small packets where the inaccuracy gap reaches more than $30 \%$ in higher delay region.

\begin{figure}[!t] 
	\centering
	\includegraphics[width=0.98\columnwidth]{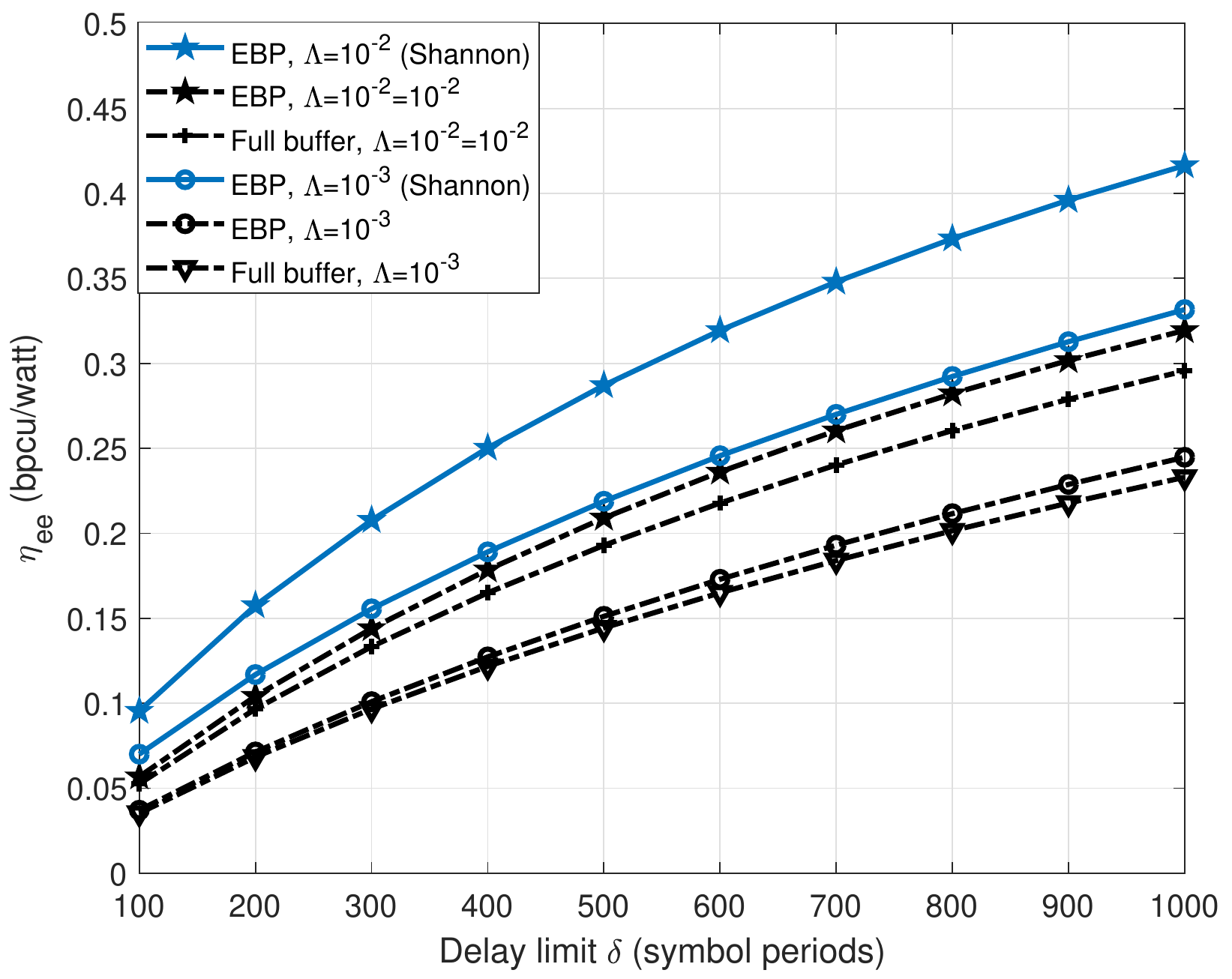}
	\centering
	\vspace{-2mm}
	\caption{EEE vs $\delta$ with and without empty buffer probability for $\Lambda=10^{-2}, 10^{-3},\ P_c=0.2 \ W, \ \zeta=1.2, \ \lambda=1$, $n=500$, and $\epsilon_t=10^{-3}$.}
	\vspace{-2mm}
	\label{EEE_Dmax}
\end{figure}

In order to present an insight about how EBP would affect the performance of multi-user network, we consider a simple exemplary setup where 2 users transmit short packets to a common BS. The BS applies successive interference cancellation where User 1 is the primary user assumed to be ultra reliable with $\epsilon_{u1}=10^{-4}$, and therefore decoded last while it transmits with higher power 6 dB. Meanwhile, User 2 is the secondary which has lower priority, where it transmits with low transmit power of 0 dB, and reliability of $\epsilon_{u1}=0.1$. Fig. \ref{multi} depicts that the NBP of User 1 does not only affect the EEE of User 1, but also affects both the EEE and NBP of User 2. 

Taking a close look at Fig. \ref{multi}, we observe that adjusting the arrival rate of User 1 to a lower level reduces its NBP probability and improves the EEE of both users. Meanwhile, the NBP probability of User 2 increases when the buffer of User 1 is more busy. This happens because User 2 suffers from excess interference from User 1, which forces User 2 into reducing its transmission rate. Hence, packets accumulate in the buffer of User 2 which in turn becomes more congested.

\begin{figure}[!t] 
	\centering
	\includegraphics[width=1\columnwidth]{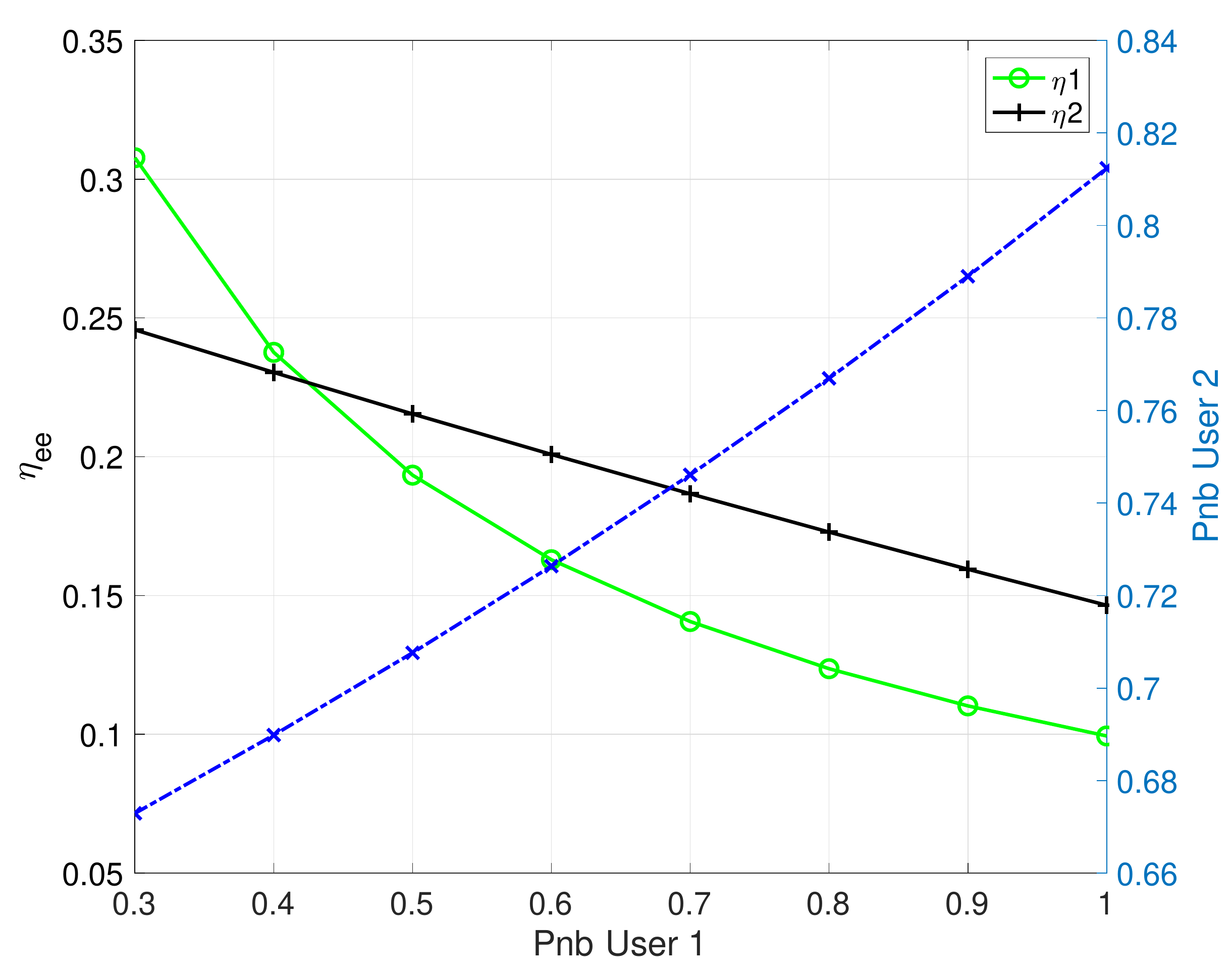}
	\centering
	\vspace{-3mm}
	\caption{Multi-user performance evaluation with EBP for $\epsilon_{u1}=10^{-4}$, $\epsilon_{u2}=0.1$, $\rho_1=6$ dB, $\rho_2=0$ dB, $\theta=0.01, P_c=0.2, \zeta=1.2, \delta=500$, $\lambda_2=0.5$ and $n=500$.}
	\vspace{-2mm}
	\label{multi}
\end{figure}

In Fig. \ref{optimal_power}, we plot the optimum power allocation for maximizing the EEE as a function of the maximum delay $\delta$ in case of EBP and always full buffer where $\rho_{max}=10$ dB. The target error outage probability is fixed at $\epsilon=10^{-4}$. The plot shows that the optimal power allocation is significantly higher when the delay outage probability $\Lambda$ is lower and when EBP is considered. The figure also depicts that Shannon's model does not render an accurate power allocation to maximize the EEE; in fact, it underestimates the optimum power allocation when compared to the finite blocklength model. The power gap ranges from $2$ to $4$ dB as shown in the figure. Thus, we can exploit the extra power allocation that results from considering empty buffer and applying the finite blocklength model in order to efficiently boost the EC. It is also observed that the optimal power allocation increases when the delay tolerance becomes higher. The intuition behind this is that when the network tolerates higher delays, it allows for improving the throughput by allocating higher power without wasting the network resources. This improvement occurs in the same way when considering empty buffer probability and when increasing the line of sight (e.g. Ricean fading where $m>1$).

\begin{figure}[!t] 
	\centering
	\includegraphics[width=1\columnwidth]{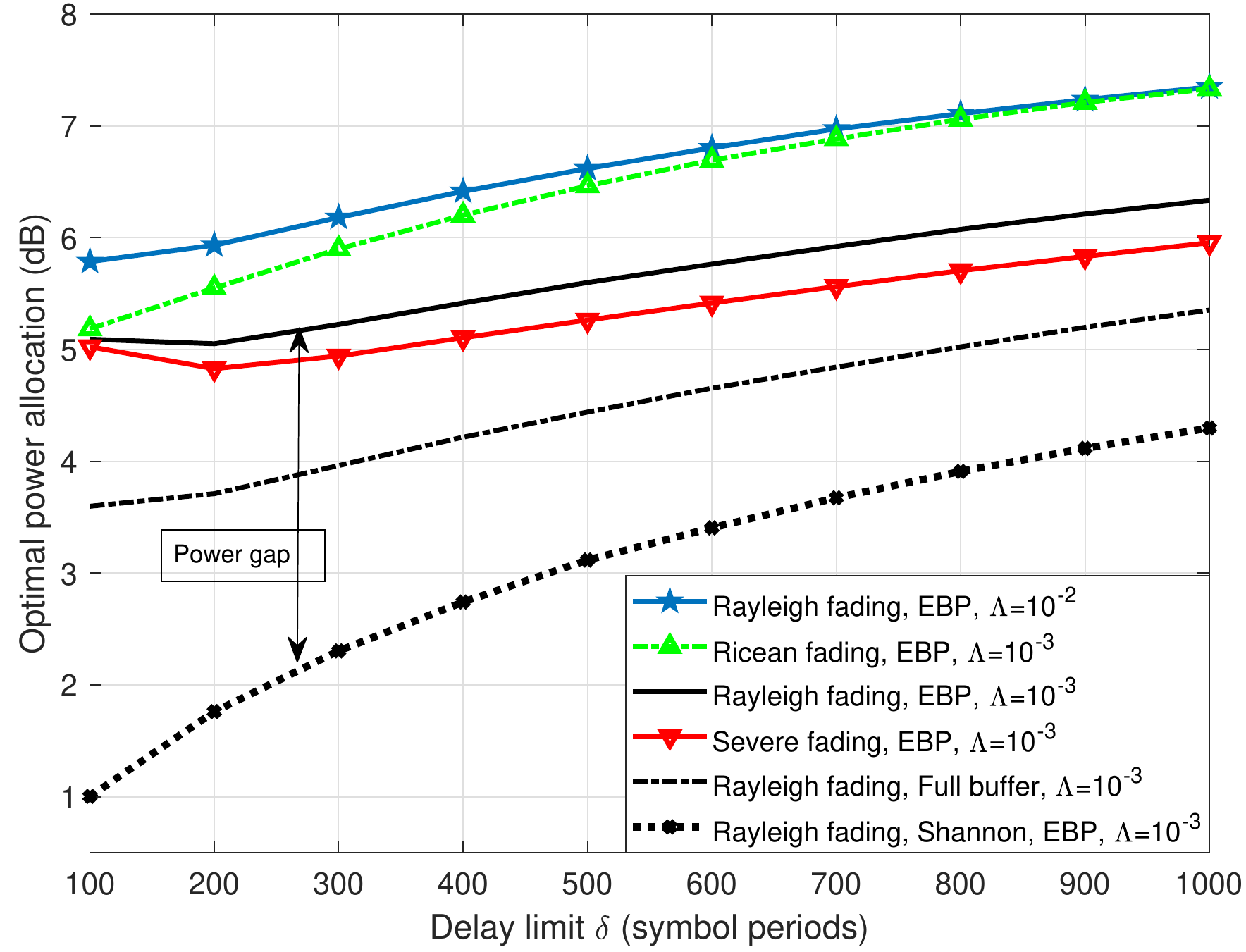}
	\centering
	\vspace{-2mm}
	\caption{Optimal power allocation vs $\delta$ with and without empty-buffer probability for $\Lambda=10^{-2}, 10^{-3}, P_c=0.2, \zeta=1.2, \lambda=1, \rho_{max}=13$ dB  and $\epsilon=10^{-4}$.}
	\label{optimal_power}
	\vspace{-2mm}
\end{figure}

In Fig. \ref{EBP_retx}, we illustrate the EEE gain of the EBP retransmission scenario. The system parameters are $\rho=6$ dB, $\epsilon=10^{-9}, \lambda=0.5$. The figure shows that our proposed EBP scheme with ARQ enhances the EEE when compared to the classical EBP and full buffer models. In fact, the EEE of EBP model with ARQ is more than double of the normal EBP case when the delay constraint is very strict and the delay exponent approaches high values at $\theta>0.1$. Thus, the EEE gain is more relevant for delay stringent networks. In this case, the EEE is upper bounded by 1.07 (bpcu/watt) as obtained from Theorem \ref{ub}. The plot also compares different error allocation strategies for the first and second transmission rounds. It is obvious that the equal error allocation is not optimal enough to maximize the EEE. However, the minimum transmit power strategy highly approaches EEE optimality.

\begin{figure}[!t] 
	\centering
	\includegraphics[width=1\columnwidth]{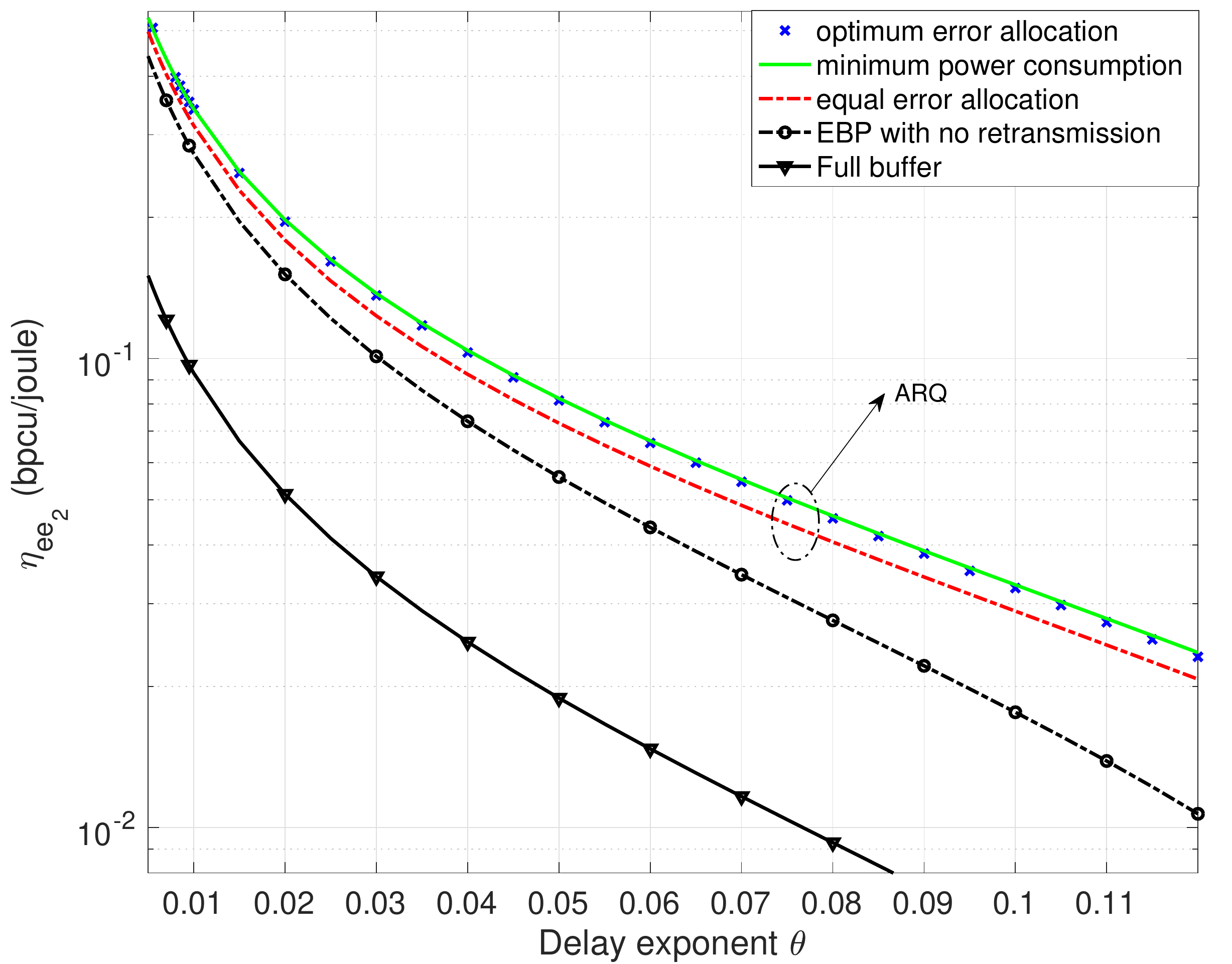}
	\vspace{-1mm}
	\caption{Effective energy efficiency $\eta_{ee}$ for different delay exponents $\theta$, where $\rho=6$ dB, $n=500,\epsilon=10^{-9}$, $P_c=0.2, \zeta=1.2, \lambda=0.5$.}
	\label{EBP_retx}
	\vspace{-1mm}
\end{figure}

In Fig. \ref{pc}, we depict the total power consumption accompanied by each scenario as function of the arrival rate. The network parameters are $\rho=6$ dB, $n=500,\epsilon=10^{-9}, \theta=0.01$, where $\lambda$ is varied this time as shown on the figure axis. Although this effect is marginal for the full buffer model, the figure depicts that higher arrival rates consume more power as there are more packets to transmit. However, the power consumption is lower for the EBP model. Despite retransmissions which consume high power, the EBP model with retransmissions is the most power saving scheme, since the packets are transmitted with higher error probabilities for each single transmission which boosts the service rate and reduces the traffic congestion at the buffer and hence, the power consumption of the whole network.

\begin{figure}[!t] 
	\centering
	\includegraphics[width=0.95\columnwidth]{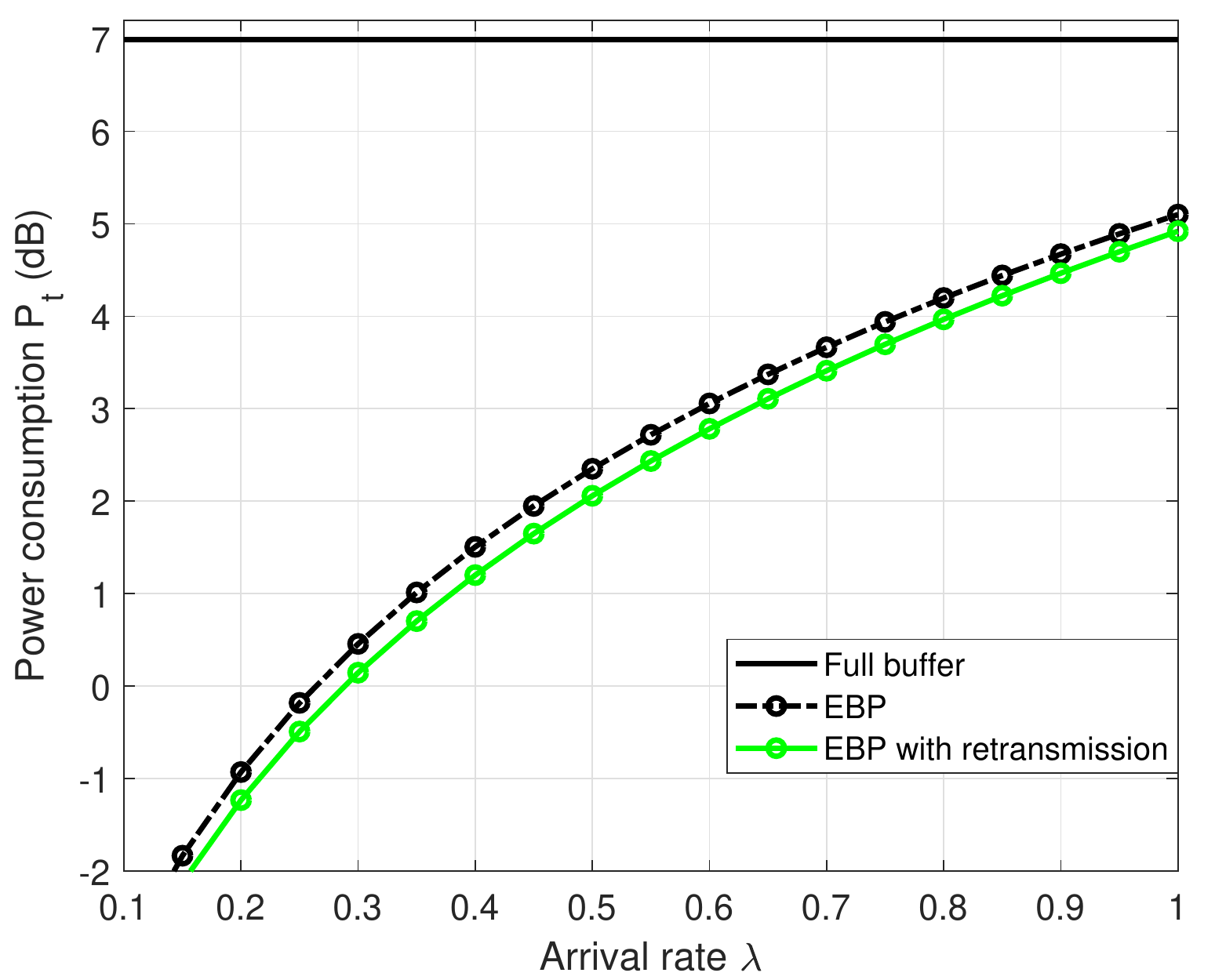}
	\vspace{-1mm}
	\caption{Power consumption $P_t$ for different arrival rates $\lambda$, where $\rho=6$ dB, $n=500,\epsilon=10^{-9}$, $P_c=0.2, \zeta=1.2, \theta=0.01$.}
	\label{pc}
	\vspace{-1mm}
\end{figure}

Finally, Fig. \ref{average_delay} illustrates the normalized delay $\tau_n$ as a function of arrival rate $\lambda$ for the EBP-ARQ scheme with two transmissions. The figure shows that the normalized delay $\tau_n$ diminishes as the network traffic becomes higher. This is because, when the queue becomes more congested, there are fewer chances for the buffer to become empty and hence, the opportunity for a second transmission disappears. Hence, the delay becomes only one transmission delay which is lower than the delay in case of two transmissions. It is noted that the delay becomes worse for lower reliability requirement as the error is also relaxed in the first transmission which leads to higher probability of occurrence for the second transmission and longer delay. 

Moreover, for the same reliability constraint, boosting the transmit power does not reduce latency. This is because for this high power, if the network becomes more congested, it slightly affects the non-EBP which maintains its low value and allows for second transmission which causes longer delay. The delay in its worst case is still only 3\% higher than the delay of one transmission scheme. Hence, we obtain significantly higher EEE with only limited increase in the delay, while maintaining reliability at the same level.
\begin{figure}[!t] 
	\centering
	\includegraphics[width=1\columnwidth]{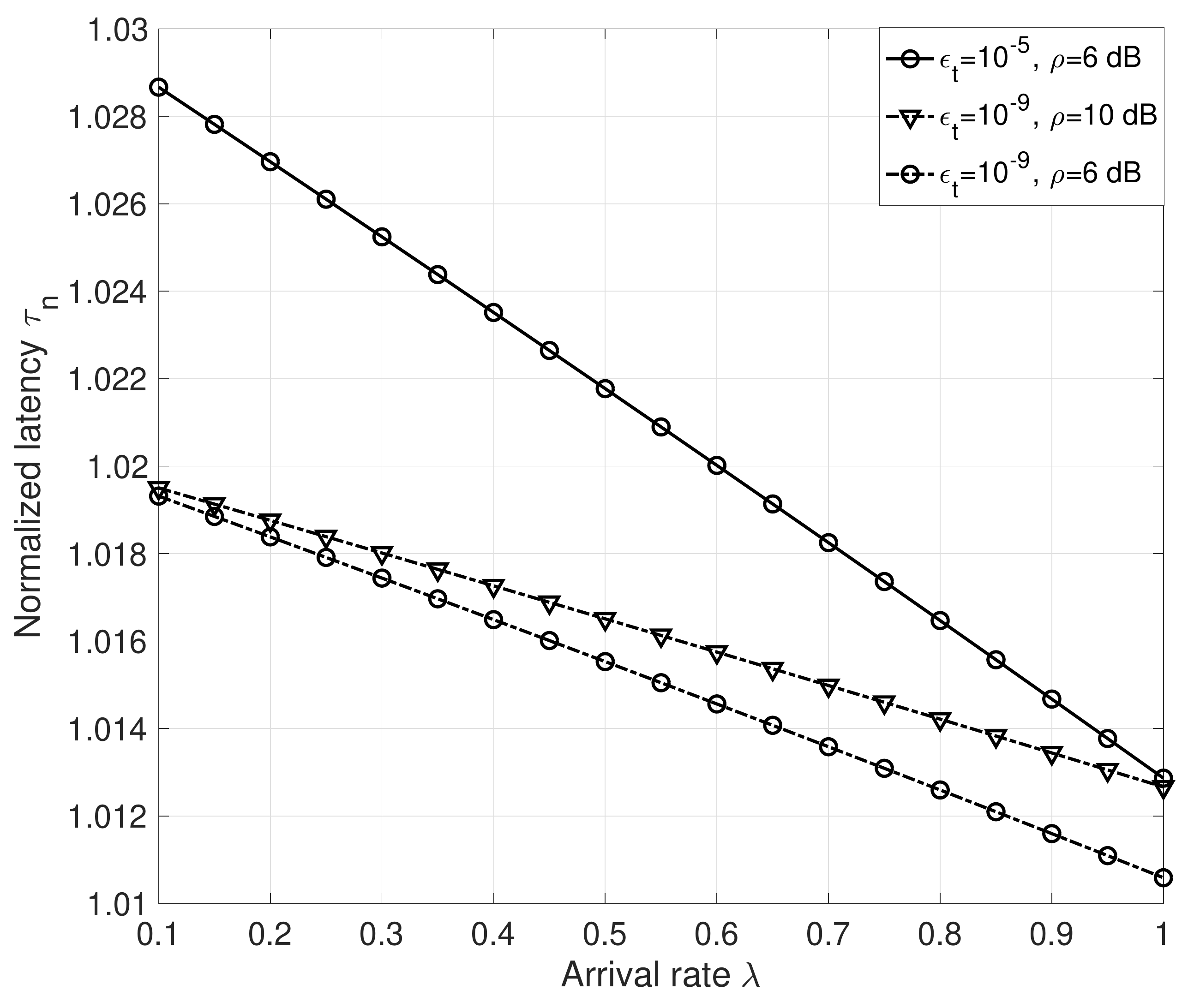}
	\vspace{-1mm}
	\caption{Normalized delay $\tau_n$  for different arrival rates $\lambda$, where $n=500, P_c=0.2, \zeta=1.2, \theta=0.01$.}
	\label{average_delay}
	\vspace{-1mm}
\end{figure}

\section{Conclusions} \label{conclusion}
In this work, we presented a detailed analysis of the EEE for ultra reliable delay constrained networks in the finite blocklength regime. For Nakagami-$m$ quasi-static fading channels, we proposed an approximation for the EC. Then we characterized the EEE maximizers in terms of optimum error probability and power allocation for the Rayleigh fading case. The results revealed that Shannon's model overestimates the EEE and underestimates the optimum power allocation when compared to the finite blocklength model. Further results indicated that allowing for larger delays significantly boosts EEE. We showed that the advantage of considering non-empty buffer probability and flexible transmission power is twofold since it significantly improves the EEE of networks operating in the finite blocklength regime and allows for retransmission of faulty packets with a significant boost in the EEE, reliability, and limited delay extension. For the EBP retransmission scenario, we derived the upper bound of the EEE and provided a low complexity solution for the optimization problem of maximizing the EEE and minimizing the power consumption. The solution showed that EBP model with retransmission is an ultra-reliable power saving scheme which improves the energy efficiency with limited increase in latency. Better performance and higher EEE gain could be achieved by applying Chase Combining (CC-HARQ) and Incremental redundancy (IR-HARQ) protocols \cite{Dosti_ARQ}. This is left as possible extension for this work along with the analysis of EEE in multi-user networks as in Fig. \ref{multi}.

\appendices 
\section{PROOF OF Lemma 1}
	Using (\ref{pdf}) in (\ref{EC}), we attain 
\begin{align}\label{EC2}
\begin{split}
C_e(\rho,\theta,\epsilon)&= \\ -\frac{1}{n\theta} \ln&\left(\frac{m^m}{\Gamma (m)}\int_{0}^{\infty}
\left( \epsilon+(1-\epsilon)e^{-\theta n r}\right)  z^{m-1} e^{-mz} dz\right).  
\end{split}
\end{align}
From (\ref{eq3}), we have		
\begin{align}\label{e1}
e^{-\theta n r}=e^{-\theta n \log_2(1+\rho z)}e^{\theta \sqrt{n(1-\frac{1}{(1+\rho z)^{2}})} Q^{-1}(\epsilon)\log_2e}, 
\end{align}	
where \vspace{-3mm}
\begin{flalign}\label{e2}
e^{-\theta n \log_2(1+\rho z)}
&=(1+\rho z)^{\alpha}
\end{flalign}
\begin{align}\label{e3}
e^{\theta \sqrt{n(1-\frac{1}{(1+\rho z)^{2}})} Q^{-1}(\epsilon)\log_2e}=e^{\beta \gamma}.  
\end{align}
We resort to the Taylor expansion to obtain $e^{cx} = \sum_{n=0}^\infty \frac{(cx)^n}{n!}$. It follows from (\ref{e1}),(\ref{e2}) and (\ref{e3}) that the expression in (\ref{EC2}) can be written as 
\begin{align}\label{general2}
\begin{split}
&EC(\rho_i,\theta,\epsilon)=-\frac{1}{n\theta} \ln\left[  \int_{0}^{\infty}\epsilon \frac{m^m}{\Gamma (m)} \  z^{m-1} e^{-mz} dz\right. \\
& \left.+(1-\epsilon)\int_{0}^{\infty} \frac{m^m}{\Gamma (m)}
(1+\rho_i z)^{\alpha}\sum_{n=0}^\infty \frac{(\beta\gamma)^n}{n!}z^{m-1} e^{-mz} dz\right].
\end{split}
\end{align}
The infinite series in (\ref{general2}) can be truncated to a finite sum of terms and we evaluate the accuracy of the expression noting that the accuracy increases with the number of terms. But, it is noticed that when testing for 
different system parameters ($N$, $\rho$, $\theta$, $n$), the accuracy for expanding 1 term is 92.7$\%$, 2 terms is 99$\%$ and 99.9$\%$ for 3 terms only. Henceforth, in our analysis, 3 terms will be enough and (\ref{general2}) reduces to (\ref{general}).
\section{PROOF OF Theorem 1}
	The coefficients of Rayleigh channel are distributed according to the probability density function (PDF) $f_Z(z)=e^{-z}$. This corresponds to a Nakagami-$m$ fading parameter $m=1$. Applying the second order Taylor expansion to obtain $e^{\beta \gamma} = 1+(\beta \gamma)+\frac{(\beta \gamma)^2}{2}$, it follows from Theorem 1 that for Rayleigh distributed channels
	\begin{align}\label{general3}
	\begin{split}
	&\psi(\rho,\theta,\epsilon)=  \epsilon 
	+(1-\epsilon)\left[ \int_{0}^{\infty}
	(1+\rho z)^{\alpha}e^{-z} \mathrm{d}z + \right. \\
	&\left. \beta\int_{0}^{\infty}
	(1+\rho z)^{\alpha} \gamma e^{-z} \mathrm{d}z + \frac{\beta^2}{2}\int_{0}^{\infty}
	(1+\rho z)^{\alpha} \gamma^2 e^{-z} \mathrm{d}z  \right].  
	\end{split}
	\end{align}
	The first integral can be written as $e^{\frac{1}{\rho}}  \rho^\alpha \Gamma\left(\alpha+1,\frac{1}{\rho} \right)$ By applying Laurent's expansion for $\gamma$ \cite{Complex_Analysis}, we obtain $\gamma\approx1-\frac{1}{2\left(1+\rho z \right)^2 }$. Hence, the second and third integrals can be written as  
	$e^{\frac{1}{\rho}} \beta \rho^\alpha  \left( \Gamma\left(\alpha+1,\frac{1}{\rho} \right)-\frac{\Gamma\left(\alpha-1,\frac{1}{\rho}\right) }{\rho^{2}}\right) $, and $e^{\frac{1}{\rho}} \frac{\beta^2}{2} \rho^\alpha  \left( \Gamma\left(\alpha+1,\frac{1}{\rho} \right)-\frac{\Gamma\left(\alpha-1,\frac{1}{\rho}\right) }{\rho^{2}}\right) $, respectively leading to (\ref{c2.2}).
	
\section{PROOF OF Theorem 2}
Since the logarithmic term is dominant in the rate equation given in (\ref{eq3}), it is quite straightforward to verify that the rate function has a negative second derivative for practical rate and SNR regions and therefore is concave in transmit power. The mathematical proof proceeds as follows. First, let $\phi=\frac{Q^{-1}(\epsilon)\log_2(e)}{\sqrt{n}}$ and note that $\phi$ should be a strictly positive parameter. Moreover, $\phi$ is less than unity for practical values of $n$ and $\epsilon$, where $n\geq1$ and $\epsilon\leq0.1$, which in fact are guaranteed for URLLC operation where n>100 and $\epsilon<10^{-4}$ \cite{PopovskiURLLC2019}. This dictates that the denominator of $\phi$ is higher than its numerator since $\sqrt{n}$ will be large enough to exceed $Q^{-1}(\epsilon)\log_2(e)$. Then from \eqref{eq3}, we have
\begin{align} \label{e4}
\frac{\partial r}{\partial \rho}=\frac{z}{ (1+\rho z) \log 2}-\frac{\phi z}{\left( 1+\rho z\right) ^3 \sqrt{ 1-\frac{1}{\left( 1+\rho z\right) ^{2}} }},
\end{align}
\begin{align} \label{e6}
\frac{\partial^2 r}{\partial \rho^2}&=\frac{3\phi z^2}{ \left( 1+\rho z\right) ^4 \sqrt{ 1-\frac{1}{\left( 1+\rho z\right) ^{2}} }}  \notag \\
&+ \frac{\phi z^2}{\left( 1+\rho z\right) ^6 \left( 1-\frac{1}{\left( 1+\rho z\right) ^{2}} \right)^\frac{3}{2}}-\frac{z^2}{\left( 1+\rho z\right) ^2 \log 2},
\end{align}
which is dominated by the negative term, since the other terms are multiplied by $\phi\lessapprox1$ and raised to a high power in the denominator, and thus vanish faster. This firmly holds for non-extremely low SNR (i.e, $\geq-10$ dB) regions. Following a similar procedure as in \cite{paper8} based on \cite{EEE_concave}, we can conclude that the EEE in the finite blocklength regime is also a quasi-concave function of power and strictly concave in its upper contour.

\section{PROOF OF Lemma 2}
	For $\rho=0$, the achievable rate $r=0$ and the numerator of (\ref{EEE0}) becomes 0. Applying L'Hopital's rule for the denominator, we have   
	\begin{align}\label{lim1}
	\begin{split}
	&\lim\limits_{\rho\rightarrow 0} \frac{\rho}{\mathop{\mathbb{E}}\left[ r\right] }= \lim\limits_{\rho\rightarrow 0}\frac{1}{\mathop{\mathbb{E}}\left[ z\left(\frac{1}{(1+\rho z) \ln 2}-\frac{Q^{-1}(\epsilon)\log_2(e)}{\sqrt{n}\left( 1+\rho z\right)^3 \gamma} \right)\right] }=0. \\
	\end{split}
	\end{align}
	Thus the denominator of (\ref{EEE0}) equals to $P_c$ yielding 0 for the EEE. 
	
	For the second condition, the numerator of (\ref{EEE0}) is upper bounded by $-\frac{\ln \epsilon }{n \theta}$, while L'Hopital's rule for the denominator, we obtain
	\begin{align}\label{lim2}
	\begin{split}
	&\lim\limits_{\rho\rightarrow \infty}\frac{1}{\mathop{\mathbb{E}}\left[ z\left(\frac{1}{(1+\rho z) \ln 2}-\frac{Q^{-1}(\epsilon)\log_2(e)}{\sqrt{n}\left( 1+\rho z\right)^3 \gamma} \right)\right] }=\infty. \\
	\end{split}
	\end{align}	
	Thus, the denominator of (\ref{EEE0}) tends to infinity which nulls the EEE. Hence, (\ref{EEE1}) holds as well under finite blocklength regime, which concludes the proof.

\section{PROOF OF Theorem 4}	

Since the denominator of the EEE does not depend on $\theta$, so the problem is to define the limits of the EC. First, we define the limit of the EC for one transmission scenario and extremely strict delay constraint, where $\theta \rightarrow \infty$ as
\begin{align} \label{e30.6}
\lim\limits_{\theta\rightarrow \infty}C_{e1}(\epsilon) =\lim\limits_{\theta\rightarrow \infty} \frac{-1}{n\theta} \ln\left(\mathop{\mathbb{E}_{Z}}\left[\epsilon+(1-\epsilon)e^{-n\theta r(\epsilon)}\right]\right)=0,
\end{align} 
which means that the EC vanishes as the delay constraint becomes infinitely strict and consequently, the EEE vanishes too. A similar procedure shows the same zero lower bound for the retransmission scenario. Next, we define the limit of the EC for loose delay constraint, where $\theta \rightarrow 0$ as follows
\begin{align} \label{e30.8}
\lim\limits_{\theta\rightarrow 0}C_{e1}(\epsilon) \!=\!\lim\limits_{\theta\rightarrow 0} -\frac{f(\theta)}{g(\theta)}\!=\!\lim\limits_{\theta\rightarrow 0}-\frac{\frac{\partial f(\theta)}{\partial \theta}}{\frac{\partial g(\theta)}{\partial \theta}}\!=\!\lim\limits_{\theta\rightarrow 0}-\frac{\frac{\partial f(\theta)}{\partial \theta}}{n},
\end{align} 
which follows from L'Hopital rule, where $f(\theta)=\ln\left(\mathop{\mathbb{E}_{Z}}\left[\epsilon+(1-\epsilon)e^{-n\theta r(\epsilon)}\right]\right)$ and $g(\theta)=n\theta$. Differentiating $f(\theta)$ with respect to $\theta$, we get
\begin{align} \label{e30.10}
\frac{\partial f(\theta)}{\partial \theta}&=\frac{\frac{\partial }{\partial \theta}\left\lbrace \mathop{\mathbb{E}_{Z}}\left[\epsilon+(1-\epsilon)e^{-n\theta r(\epsilon)}\right]\right\rbrace }{\mathop{\mathbb{E}_{Z}}\left[\epsilon+(1-\epsilon)e^{-n\theta r(\epsilon)}\right]}\notag \\
&=\frac{\frac{\partial }{\partial \theta}\int_{0}^{\infty}
	\left( \epsilon+(1-\epsilon)e^{-n \theta  r}\right)  e^{-z} dz}{\mathop{\mathbb{E}_{Z}}\left[\epsilon+(1-\epsilon)e^{-n\theta r(\epsilon)}\right]}.
\end{align} 
Applying Leibniz's rule, we obtain
\begin{align} \label{e30.12}
\frac{\partial f(\theta)}{\partial \theta}&=\frac{\int_{0}^{\infty}\frac{\partial }{\partial \theta}
	\left( \epsilon+(1-\epsilon)e^{-n \theta  r}\right)  e^{-z} dz}{\mathop{\mathbb{E}_{Z}}\left[\epsilon+(1-\epsilon)e^{-n\theta r(\epsilon)}\right]} \notag \\
&=\frac{-n(1-\epsilon)\mathop{\mathbb{E}_{Z}}\left[r e^{-n\theta r(\epsilon)}\right] }{\mathop{\mathbb{E}_{Z}}\left[\epsilon+(1-\epsilon)e^{-n\theta r(\epsilon)}\right]}.
\end{align} 
Plugging back into \eqref{e30.8}, we reach
\begin{align} \label{e30.14}
\lim\limits_{\theta\rightarrow 0}C_{e}(\epsilon) &=\lim\limits_{\theta\rightarrow 0}-\frac{\frac{-n(1-\epsilon)\mathop{\mathbb{E}_{Z}}\left[r e^{-n\theta r(\epsilon)}\right] }{\mathop{\mathbb{E}_{Z}}\left[\epsilon+(1-\epsilon)e^{-n\theta r(\epsilon)}\right]}}{n} \notag \\
&=(1-\epsilon)\mathop{\mathbb{E}_{Z}}\left[r(\epsilon)\right]= (1-\epsilon) r_0,
\end{align} 
which represents the upper bound throughput of the finite blocklength transmission when no delay constraint is imposed. Following the same procedure, we can deduce that the EC of the EBP ARQ scenario is given by the numerator of Theorem \ref{ub}.

\section{PROOF OF Theorem 5}	
	
	The proof for the pseudo-convexity of $p_{nb}^{'}$ in $\epsilon_1$ goes as follows. At the first glance, it is possible to prove that $\kappa$ is a concave function in $\epsilon_1$. First, let
	\begin{align} \label{e52}
	r(\epsilon)\approx&\log_2(1+\rho|h|^2)-\mu Q^{-1}(\epsilon) ,
	\end{align}
	where $\mu=\frac{\log_2(e)}{\sqrt{n}}\sqrt{ 1-\frac{1}{\left( 1+\rho|h|^2\right) ^{2}} }$. Applying the derivatives of the inverse Q-function from \cite{paper5}, we obtain the derivatives of the rate expectations with respect to $\epsilon_1$ as  
	\begin{align} \label{e54}
	\frac{\partial r_1}{\partial \epsilon_1}&=\mathbb{E}\left[\mu\right]  \sqrt{2\pi} e^{\frac{\left( Q^{-1}(\epsilon_1)\right) ^2 }{2}}, \\ 
	\frac{\partial^2 r_1}{\partial \epsilon_1^2}&=-\mathbb{E}\left[\mu\right]2\pi Q^{-1}(\epsilon_1)e^{\frac{\left( Q^{-1}(\epsilon_1)\right) ^2 }{2}}, \label{e56}   \\
	\frac{\partial r_2}{\partial \epsilon_1}&=-\frac{\epsilon}{\epsilon_1^2}\mathbb{E}\left[\mu\right]  \sqrt{2\pi} e^{\frac{\left( Q^{-1}(\epsilon_1)\right) ^2 }{2}}, \label{e58}  \\
	\frac{\partial^2 r_2}{\partial \epsilon_1^2}&=\frac{\epsilon}{\epsilon_1^2}\mathbb{E}\left[\mu\right]2\pi Q^{-1}(\epsilon_1)e^{\frac{\left( Q^{-1}(\epsilon_1)\right) ^2 }{2}}  \notag \\
	&+\mathbb{E}\left[\mu\right]  \sqrt{2\pi} e^{\frac{\left( Q^{-1}(\epsilon_1)\right) ^2 }{2}}\frac{2\epsilon}{\epsilon_1^3} \notag \\
	&=-\frac{\epsilon}{\epsilon_1^2} \frac{\partial^2 r_1}{\partial \epsilon_1^2} +\mathbb{E}\left[\mu\right]  \sqrt{2\pi} e^{\frac{\left( Q^{-1}(\epsilon_1)\right) ^2 }{2}}\frac{2\epsilon}{\epsilon_1^3}. \label{e60}
	\end{align}
	Then, we obtain the second derivative of $\kappa$ w.r.t $\epsilon_1$ as follows 
	\begin{align} \label{e61}
	\frac{\partial \kappa}{\partial \epsilon_1}=-r_1+\frac{\partial r_1}{\partial \epsilon_1}(1-\epsilon_1)+\frac{1}{2}\left(\epsilon_1\frac{\partial r_2}{\partial \epsilon_1}+r_2\right) 
	\end{align} \vspace{-6mm}
	\begin{align} \label{e62}
	\frac{\partial^2 \kappa}{\partial \epsilon_1^2}&=-2\frac{\partial r_1}{\partial \epsilon_1}+\frac{\partial^2 r_1}{\partial \epsilon_1^2}(1-\epsilon_1)+\frac{1}{2}\left(2\frac{\partial r_2}{\partial \epsilon_1}+\epsilon_1 \frac{\partial^2 r_2}{\partial \epsilon_1^2}\right)= \notag \\
	-2&\frac{\partial r_1}{\partial \epsilon_1}+\frac{\partial^2 r_1}{\partial \epsilon_1^2}(1-\epsilon_1)+\frac{1}{2}\left\lbrace \cancel{-2\frac{\epsilon}{\epsilon_1^2}\mathbb{E}\left[\mu\right]\sqrt{2\pi} e^{\frac{\left( Q^{-1}(\epsilon_1)\right) ^2 }{2}}}\right. \notag \\
	&\left. -\frac{\epsilon}{\epsilon_1} \frac{\partial^2 r_1}{\partial \epsilon_1^2} +\cancel{\epsilon_1\mathbb{E}\left[\mu\right]  \sqrt{2\pi} e^{\frac{\left( Q^{-1}(\epsilon_1)\right) ^2 }{2}}\frac{2\epsilon}{\epsilon_1^3}} \right\rbrace \notag \\
	&=-2\frac{\partial r_1}{\partial \epsilon_1}+\frac{\partial^2 r_1}{\partial \epsilon_1^2}(1-\epsilon_1-\frac{\epsilon}{2\epsilon_1}).
	\end{align}
	Note that practically, the term $\frac{\epsilon}{2\epsilon_1}$ is close to zero since $\epsilon_1>>\epsilon$, and hence, $(1-\epsilon_1)>\frac{\epsilon}{2\epsilon_1}$. Since the derivatives in \eqref{e54} and \eqref{e56} are strictly negative, we can deduce that the second derivative of $\kappa$ in \eqref{e62} renders a strictly negative value. Therefore, $\kappa$ is strictly concave in $\epsilon_1$. Back to equation \eqref{e78}, the denominator of the non-empty buffer probability is a linear decreasing function of $\kappa$, which indicates that the denominator is a convex function in $\epsilon_1$. The numerator can be proven to be a convex decreasing function of $T$ which means that the numerator is a convex function in $\epsilon_1$. First we obtain the derivatives of the numerator of $p_{nb}^{'}$ with respect to $T$ as follows
	
	\begin{align} \label{e14}
	\frac{\partial \ num\left\lbrace p_{nb}^{'}\right\rbrace }{\partial \kappa} &=\frac{2\kappa-4\lambda}{2\sqrt{\kappa^2+4\lambda(r_o-\kappa)}}-1, \\
	\frac{\partial^2 \ num\left\lbrace p_{nb}^{'}\right\rbrace }{\partial \kappa^2} &=\frac{4\lambda(r_o-\lambda)}{\left( \kappa^2+4\lambda(r_o-\kappa)\right) ^\frac{3}{2}},
	\end{align}
	which gives a strictly positive value and hence, the numerator of of $p_{nb}^{'}$ is convex in $\kappa$. Since the numerator of $p_{nb}^{'}$ is a non-increasing convex function in $\kappa$ and $\kappa$ is a concave function of $\epsilon_1$, it follows from \cite{Boyd} that it is a convex function in $\epsilon_1$, while the numerator is also a negative function. Hence and according to Proposition 2.9 in \cite{EEEbook}, minimizing $p_{nb}^{'}$ is a psuedo-convex fractional program. 

\bibliographystyle{IEEEtran}
\bibliography{di}
\end{document}